\documentclass[12pt]{article}

\usepackage{PRIMEarxiv}
\usepackage[utf8]{inputenc} 
\usepackage{setspace}      
\usepackage[T1]{fontenc}    
\usepackage{lmodern}  
\usepackage{hyperref}       
\usepackage{url}            
\usepackage{booktabs, multicol, multirow}       
\usepackage[table]{xcolor}
\usepackage{natbib}
\usepackage{amsfonts, amsmath, amssymb, amsthm, bm}       %
\usepackage{nicefrac}      
\usepackage{microtype}      
\usepackage{lipsum}
\usepackage{fancyhdr}       
\usepackage{graphicx}       
\graphicspath{{media/}}     
\usepackage{algorithm, algpseudocode}

\newcommand{\indep}{\perp \!\!\! \perp}

\DeclareMathOperator*{\argmax}{arg\,max}
\DeclareMathOperator*{\argmin}{arg\,min}

\newcommand{\B}{{\bf B}}

\newcommand{\bb}{{\bf b}}

\newcommand{\M}{{\bf M}}

\newcommand{\bv}{{\bm v}}
\newcommand{\bu}{{\bm u}}
\newcommand{\X}{{\bf X}}
\newcommand{\x}{{\bm{x}}}
\newcommand{\Z}{{\bf Z}}

\newcommand{\bbeta}{\bm{\beta}}
\newcommand{\bEta}{\bm{\eta}}

\newcommand{\bLambda}{\bm{\Lambda}}

\newcommand{\R}{\mathbb{R}}
\newcommand{\spc}{{\mathcal S}_{Y|\X}}
\newcommand{\E}{\mathbb{E}}

\newcommand{\Sig}{\bm{\Sigma}}

\newtheorem{theorem}{Theorem}

\newtheorem{proposition}{Proposition}
\theoremstyle{definition}

\newtheorem{assumption}{Assumption}

\newtheorem{remark}{Remark}

\title{Robust Dual-Regularized Variable Selection under Outlier Contamination}

\author{
  Abdul-Nasah Soale \thanks{Corresponding author}\\
  Department of Mathematics, Applied Mathematics, and Statistics, \\
Case Western Reserve University, Cleveland, OH, USA\\
  \texttt{abdul-nasah.soale@case.edu} \\
   \And
 Adewale F. Lukman \\
 Department of Mathematics and Statistics,\\
 University of North Dakota,  Grand Forks, ND 58202, USA\\
 \texttt{adewale.lukman@und.edu} \\
  \AND
 Essoham Ali\\
 CNRS UMR 6205, Univ Bretagne Sud, Vannes, France\\
 \texttt{essoham.ali@univ-ubs.fr} 
}

\begin{document}
\maketitle

\begin{abstract}
Real data often contain unusual observations that can exert disproportionate effects on variable selection, especially in complex predictor settings. We propose a two-stage {\it sparse median outer product of gradients (smOPG)} method for variable selection in single index models with outlier contamination. We first estimate sparse local gradients via \(\ell_1\)-penalized local median regression and then recover the active predictor set from a rank-one sparse approximation of the resulting gradient matrix using regularized singular value decomposition. The combination of median regression and local weighting provides robustness to both response outliers and leverage points. Extensive simulations across varying dimensions and contamination mechanisms demonstrate the favorable variable selection performance of smOPG relative to existing methods. Applications to air pollution and genomic data demonstrate practical utility, while theory establishes active-set recovery without requiring selection consistency of individual local regressions.
\end{abstract}


\section{Introduction}\label{sec:intro}
Variable selection is a central problem in high-dimensional regression, especially in applications from medicine, economics, and biology, where interpretability is often as important as predictive accuracy. In many such settings, fully parametric models are too restrictive, whereas fully nonparametric methods suffer from severe dimensionality constraints. single index models provide a useful compromise by allowing the response to depend on the predictors only through an unknown linear combination \(\bbeta^\top \X\), thereby capturing nonlinear regression structure while retaining a low-dimensional representation \citep{ichimura1993semiparametric, hardle1993optimal}.

When \(\bbeta\) is sparse, recovery of its support yields a natural notion of active subset selection. This perspective connects sparse single index modeling with sufficient dimension reduction (SDR), where the objective is to identify a low-dimensional subspace of the predictor space that preserves all regression information about the response \citep{li1991sliced, xia2002adaptive, li2018sufficient}. In the single index setting, this subspace is one-dimensional and coincides with \(\mathrm{span}(\bbeta)\). Variable selection may therefore be viewed as recovery of a sparse sufficient direction \citep{wang2008nonlinear}.

A substantial literature studies estimation and variable selection for single index models and related dimension reduction problems. Some methods estimate the index direction directly without specifying the link function \citep{powell1989semiparametric, hardle1989investigating, li1989regression, li1991sliced}, whereas others estimate the index and the link jointly through semiparametric or nonparametric smoothing \citep{hardle1993optimal, ichimura1993semiparametric, xia2002adaptive, xia2006asymptotic, yin2005direction}. In high-dimensional settings, sparse versions of these procedures have been proposed to improve interpretability and recover active predictors.

Despite their flexibility, existing procedures can be sensitive to contamination. In practice, high-dimensional data often contain heavy-tailed response errors, atypical covariate values (leverage points), and strong predictor dependence. Mean-based methods are particularly vulnerable to extreme response values \citep{Cook1977, rajaratnam2019influence}, while local smoothing procedures may still be unstable in the presence of influential observations \citep{soale2025adaptive}. Several robust alternatives have therefore been proposed, including rank-based, quantile-based, and signed-rank methods \citep{kong2007variable, lin2019sparse, bindele2019robust}. These methods improve robustness relative to mean-based approaches, but they also have important limitations. Rank-based procedures such as those proposed in \cite{rejchel2020rank} rely on response ordering and may lose directional information when the regression link is symmetric or nonmonotone. \cite{song2022robust} proposed a composite-quantile method that utilizes information from multiple parts of the conditional distribution, which can improve efficiency but may also inherit instability from noncentral quantiles under heavy-tailed errors or contamination.

In this paper, we propose a sparse median outer-product-of-gradients procedure, denoted smOPG, for active subset recovery in single index models. We formulate the active subset recovery through a sparse sufficient dimension reduction framework based on the conditional median. First, we estimate sparse local median gradients and then aggregate them through a sparse rank-one approximation of the resulting gradient matrix. Thus, the procedure combines the robustness of the median to extreme response outliers with local kernel weighting, which downweight distant leverage points. We also establish theoretical properties of the proposed estimator, including conditions under which the second-stage aggregation recovers the active predictor set even when individual local penalized regressions are not selection consistent. 

The remainder of the article is organized as follows. Section~\ref{sec:rankfail} discusses the limitations of mean-, rank-, and composite quantile-based approaches and motivates median-based sufficient dimension reduction. Sections~\ref{sec:smOPG-LASSO} and \ref{sec:sample_est} introduce the proposed smOPG procedure at the population and sample levels, respectively. Section~\ref{sec:theory} presents the theoretical properties of the estimator. Section~\ref{sec:simulation} reports simulation studies, Section~\ref{sec:realdata} presents two real-data applications, and the conclusion is given in Section~\ref{sec:conclusion}. Proofs and additional technical results are given in the supplemental Appendix.

Throughout the paper, \(\|\cdot\|_q\) denotes the \(\ell_q\)-norm, \(|A|\) denotes the cardinality of a set \(A\), and \(\mathrm{span}(\M)\) denotes the column space of matrix \(\M\).

\section{The Case for Median-Based Variable Selection}\label{sec:rankfail}
To motivate variable selection via median sufficient dimension reduction, it is helpful to contrast it with several common alternatives, including mean-based, rank-based, and composite-quantile-based approaches. We begin with the familiar mean-regression single index model
\begin{align}
Y = g(\bbeta^\top \X) + \epsilon,
\text{ where } \E(\epsilon\mid \X)=0,
\label{eq:mean-sim}
\end{align}
so that $\E(Y\mid \X=\x)=g(\bbeta^\top \x)$.
A popular local approximation of \(g(\bbeta^\top x)\) is the kernel-weighted average
\[
\frac{\E\{Y K_h(\bbeta^\top \X-\bbeta^\top \x)\}}
{\E\{K_h(\bbeta^\top \X-\bbeta^\top \x)\}},
\]
where \(K_h(u)=h^{-1}K(u/h)\) is a univariate kernel with bandwidth \(h > 0\), and the corresponding influence function is proportional to $K_h(\bbeta^\top \X-\bbeta^\top \x)\,\epsilon$.

Thus, although kernel localization downweights observations that are far from \(\x\) in the projected index direction, the effect of the response residual remains unbounded. An observation with an unusually large response can therefore exert a disproportionate influence on the local fit and, consequently, on estimation of \(\bbeta\). This issue is particularly relevant in local smoothing, where the effective sample size is determined by the bandwidth and is typically much smaller than the full sample size. In addition, mean-based methods require existence of the conditional first moment, which may fail under sufficiently heavy-tailed error distributions, such as Cauchy.

A natural way to reduce sensitivity to extreme response magnitudes is to replace the response by ordering information. Following \cite{rejchel2020rank}, consider the rank-type transformation \(R=F(Y)\), where \(F\) denotes the distribution function. Since rank-based procedures depend on the relative ordering of responses rather than on their magnitudes, they are substantially less sensitive than mean regression to extreme response values. However, this robustness depends on whether the ordering of \(Y\) continues to reflect the underlying index structure, which can be restrictive when the regression link is not monotone. 

For example, if \(g(t)=t^2\), then even in the absence of noise the ordering of the responses is governed by \((\bbeta^\top \X)^2\), not by the signed index \(\bbeta^\top \X\). In such cases, ordering information alone does not identify the direction of \(\bbeta\), although the one-dimensional subspace \(\mathrm{span}(\bbeta)\) may still be recoverable. For this reason, many rank-based procedures rely on monotonicity assumptions on \(g\); see, for example, \cite{rejchel2020rank}. The difficulty can become more pronounced under heteroscedasticity. For example, if $Y=g(\bbeta^\top \X)+\sigma(\X)\epsilon$, then the conditional behavior of \(F(Y)\) generally depends not only on \(\bbeta^\top \X\) but also on the scale function \(\sigma(\X)\). Thus, the ordering of the responses is not governed solely by the sufficient index. Rank-based regression therefore attenuates sensitivity to response magnitude, but its success depends critically on preservation of the ordering induced by the latent index.

Quantile-based methods provide another robust alternative. In related work for single index models, \cite{song2022robust} proposed a penalized composite quantile regression approach for robust variable selection in high-dimensional settings. Their approach incorporates information across multiple quantile levels, improving robustness relative to least-squares procedures while utilizing different parts of the conditional response distribution. However, this broader use of quantile information also distinguishes them from a median-only procedure. Because composite quantile regression aggregates information across several quantile levels, its performance may be influenced by variability in noncentral quantiles, especially under heavy-tailed errors or contamination. Thus, while composite quantile methods can be more informative than a single-quantile procedure, they may also inherit instability from portions of the conditional distribution away from the center.

On the other hand, median regression retains the original response scale while limiting the influence of large residuals through a bounded score function, providing a robust alternative to mean, rank, and composite quantile regression. Specifically, under the single index model
\begin{align}
Y = g(\bbeta^\top \X)+\epsilon,
\text{ where } Q_{1/2}(\epsilon\mid \X)=0,
\label{eq:median-sim}
\end{align}
the conditional median $Q_{1/2}(Y\mid \X)=g(\bbeta^\top \X)$ with the corresponding
influence function proportional to 
\[
K_h(\bbeta^\top \X-\bbeta^\top \x)\,\psi_{1/2}(\epsilon),
\text{ where }
\psi_{1/2}(u)=\frac12-\mathbb{I}(u<0).
\]
Because \(\psi_{1/2}(u)\in[-1/2,1/2]\), the response-side contribution of any single observation is uniformly bounded, regardless of the magnitude of its residual. Median regression therefore provides substantially stronger protection against response outliers and heavy-tailed errors than mean regression, while avoiding the global ordering assumptions often needed in rank-based methods. Relative to composite quantile procedures, it focuses exclusively on the central conditional quantile and thus avoids direct reliance on tail or other noncentral quantiles. It also requires weaker moment conditions, since the conditional median remains well-defined even when the conditional mean does not exist. These features make it a natural basis for robust sufficient dimension reduction and variable selection.

\section{Active set selection via a sparse central direction}\label{sec:smOPG-LASSO}
To avoid imposing structural restrictions on the link function \(g(\cdot)\) in Model \eqref{eq:median-sim}, we formulate variable selection through sufficient dimension reduction. Specifically, we seek a direction \(\Gamma\in\mathbb R^p\) such that
\begin{align}
Y \indep \X \mid \Gamma^\top \X,
\label{sdr}
\end{align}
with \(\|\Gamma\|_2=1\) imposed for identifiability. For the single index context, the central subspace is one-dimensional, and the sufficient direction \(\Gamma\) is proportional to the index coefficient \(\bbeta\). Hence, \(\mathrm{span}(\Gamma)=\mathrm{span}(\bbeta)\), and the support of \(\Gamma\) coincides with that of \(\bbeta\).

Let \(\mathcal S_{Y\mid \X}\) denote the central subspace, i.e., the intersection of all subspaces satisfying \eqref{sdr}. Under some mild conditions, \(\mathcal S_{Y\mid \X}\) exists and is unique \citep{yin2008successive}. In the single index framework, \(\dim(\mathcal S_{Y\mid \X})=1\) and
$\mathcal S_{Y\mid \X}=\mathrm{span}(\Gamma)=\mathrm{span}(\bbeta)$.

Now, let \(\mathcal A=\{j:\Gamma_j\neq 0\}\) denote the active set, with \(s=|\mathcal A|\ll p\), and write \(\mathcal A^c=\{1,\ldots,p\}\setminus\mathcal A\). After reordering coordinates if necessary, partition
\[
\Gamma= (\Gamma_{\mathcal A}^\top, \Gamma_{\mathcal A^c}^\top)^\top,
\quad \Gamma_{\mathcal A}\in\mathbb R^s, \quad \Gamma_{\mathcal A^c}=\bm 0\in\mathbb R^{p-s},
\]
and correspondingly \(\X=(\X_{\mathcal A}^\top,\X_{\mathcal A^c}^\top)^\top\). Since \(\Gamma^\top \X=\Gamma_{\mathcal A}^\top \X_{\mathcal A}\), condition \eqref{sdr} implies
\begin{align}
Y\indep \X_{\mathcal A^c}\mid \X_{\mathcal A}.
\label{sdr_vs}
\end{align}
Therefore, \(\X_{\mathcal A^c}\) contains no additional information about the conditional distribution of \(Y\) beyond \(\X_{\mathcal A}\), and variable selection may be viewed as recovery of the support of the central direction.

Notice that the conditional independence relation in \eqref{sdr} also implies
\begin{align}
F_{Y\mid\X}(y\mid \x)=F_{Y\mid \Gamma^\top \X}(y\mid \Gamma^\top \x),
\label{cdf_sdr}
\end{align}
where \(F_{Y\mid\X}(y\mid\x)=\Pr(Y\le y\mid \X=\x)\). Hence, for every quantile level \(\tau\in(0,1)\),
\begin{align}
Q_\tau(Y\mid \X=\x)=Q_\tau(Y\mid \Gamma^\top\X=\Gamma^\top\x),
\label{quantile_sdr}
\end{align}
where $Q_\tau(Y\mid \X=\x) = \min\{y: F_{Y\mid\X}(y\mid\x) \geq \tau \}$.
Thus, each conditional quantile depends on \(\X\) only through the sufficient index. This observation underlies quantile-based dimension reduction methods such as those of \cite{Kong2014quantile} and \cite{kim2019quantile}, as well as the penalized composite quantile approach of \cite{song2022robust} for robust variable selection in high-dimensional single index models.

In this paper, we focus on the conditional median, corresponding to \(\tau=1/2\). Restricting attention to the central quantile avoids direct reliance on tail quantiles, reduces computational complexity, and yields a bounded score function, which provides protection against extreme response residuals. More formally, we work under the single index model
\begin{align}
Y=g(\bbeta^\top \X)+\epsilon
\label{median_model}
\end{align}
that satisfies the following assumptions.

\begin{assumption}\label{ass_mederror}
The conditional median of the error, \(Q_{1/2}(\epsilon\mid \X)=0\), almost surely.
\end{assumption}

\noindent Assumption \ref{ass_mederror} allows the distribution of \(\epsilon\) to depend on \(\X\) and does not require existence of \(\E(\epsilon\mid \X)\). In particular, it accommodates heavy-tailed errors for which the conditional mean may be undefined. It also allows location--scale models of the form
$Y=g(\bbeta^\top \X)+\sigma(\X)\epsilon$, provided that \(Q_{1/2}(\epsilon\mid \X)=0\) and \(\sigma(\X)>0\). In this case, $Q_{1/2}(Y\mid \X)=g(\bbeta^\top \X)$, even when the conditional scale varies with \(\X\).

\begin{assumption}\label{ass_diff}
Let \(m(\x)=Q_{1/2}(Y\mid \X=\x)\). The function \(m(\x)\) is continuously differentiable on $\mathrm{Supp}(\X)$, the support of \(\X\). In addition, the conditional density \(f_{Y\mid \X}(y\mid \x)\) exists and is continuous in a neighborhood of \(y=m(\x)\); and it is uniformly bounded away from zero on  \(\mathrm{Supp}(\X)\). 
\end{assumption}

\noindent Assumption \ref{ass_diff} is a standard local regularity condition for median regression. The positivity of the conditional density at the median ensures local identification, while differentiability of \(m(\cdot)\) permits a first-order local approximation. Under Model \eqref{median_model} and Assumption \ref{ass_mederror},
\begin{align}
m(\x)=Q_{1/2}(Y\mid \X=\x)=g(\bbeta^\top \x),
\label{median_index}
\end{align}
and therefore
\begin{align}
\nabla m(\x)=g'(\bbeta^\top \x)\bbeta.
\label{median_gradient}
\end{align}
Hence, every population median gradient lies in the central subspace, and whenever \(g'(\bbeta^\top \x)\neq 0\), it has the same support as \(\bbeta\).

To estimate the local median gradient, fix \(\x_0\) and consider the population local linear median criterion
\begin{align}
\mathcal L_h(a,\bb;\x_0)
=
\E\left[
K_h(\X-\x_0)\left|Y-a-\bb^\top(\X-\x_0)\right|
\right],
\label{population_loss}
\end{align}
where \(K_h(\cdot)\) is a \(p\)-dimensional kernel with bandwidth \(h\). Here, we adopt a full space kernel which is potentially more resistant to leverage-type contamination compared to the index-based localization. Let \((a_0,\bb_0)\) denote any minimizer of \(\mathcal L_h(a,\bb;\x_0)\), i.e.,
\begin{align}
(a_0,\bb_0)\in \argmin_{a,\bb}\mathcal L_h(a,\bb;\x_0).
\label{population_median}
\end{align}
The corresponding \(\ell_1\)-penalized version is defined by
\begin{align}
(a_{0,\lambda},\bb_{0,\lambda})
\in
\argmin_{a,\bb}
\left\{
\mathcal L_h(a,\bb;\x_0)+\lambda\|\bb\|_1
\right\}.
\label{population_sparse_median}
\end{align}

\begin{remark}
A full $p$-dimensional kernel defined on \(\X\) downweights observations that are far from the target point in any predictor direction whereas the index-based kernel controls proximity only along \(\bbeta^\top\X\). 
\end{remark}
For example, suppose \(\beta=(1,0)^\top\), \(\x_0=(0,0)^\top\), and consider an observation \(X=(0,100)^\top\). Then \(\beta^\top X=\beta^\top \x_0=0\), so an index-based kernel assigns this observation a large weight even though it is extremely far from \(\x_0\) in the second coordinate. By contrast, a kernel defined on the entire \( \X\) space  assigns it negligible weight. Thus, full space localization can reduce the influence of leverage points that are not visible after projection onto the index.

\begin{theorem}\label{thm:sparse_median}
Suppose Assumptions \ref{ass_mederror}--\ref{ass_diff} hold, and \(g\) is continuously differentiable in a neighborhood of \(\bbeta^\top \x_0\). 
Then, as \(h\to 0\), 
\[
\bb_0
=
\nabla m(\x_0)+o(1)
=
g'(\bbeta^\top \x_0)\bbeta+o(1).
\]
Consequently,
\[
\mathrm{dist}\{\bb_0,\spc\}
=
\inf_{\mathbf v\in\spc}\|\bb_0-\mathbf v\|_2
\to 0.
\]
If in addition,  \(\bbeta\) is \(s\)-sparse with active set \(\mathcal A=\{j:\beta_j\neq 0\}\), then
$\mathrm{supp}\{\nabla m(\x_0)\}\subseteq \mathcal A$, with equality whenever \(g'(\bbeta^\top \x_0)\neq 0\).
\end{theorem}

\noindent Theorem~\ref{thm:sparse_median} shows that local median regression identifies the sufficient direction through the gradient of the conditional median. Penalization is then used to exploit sparsity in that direction, rather than to create the dimension reduction structure itself.

To aggregate the information in the local median gradients, define the population outer-product-of-gradients matrix
\begin{align}
\bLambda
=
\E\left\{\nabla m(\X)\nabla m(\X)^\top\right\}.
\label{population_opg}
\end{align}
By \eqref{median_gradient},
\begin{align}
\bLambda
=
\E\left[\{g'(\bbeta^\top \X)\}^2\right]\bbeta\bbeta^\top.
\label{population_opg_rank1}
\end{align}
If $\E\left[\{g'(\bbeta^\top \X)\}^2\right]>0$, then \(\bLambda\) has rank one and $\mathrm{span}(\bLambda)=\mathrm{span}(\bbeta)=\mathcal S_{Y\mid \X}$. Its leading eigenvector is $\Gamma_0=\bbeta/\|\bbeta\|_2$ (unique up to sign). Moreover, if \(\bbeta\) is \(s\)-sparse with active set \(\mathcal A\), then
\[
\Lambda_{jk}=0
\quad\text{whenever } j\notin \mathcal A \text{ or } k\notin \mathcal A,
\]
so \(\bLambda\) is supported on \(\mathcal A\times \mathcal A\). Thus, in the single index setting, recovering the support of the leading eigenvector of \(\bLambda\) is equivalent to recovering the active set.

The sparse central direction may therefore be characterized through the constrained Rayleigh quotient
\begin{align}
\Gamma_0 = \argmax_{\|\Gamma\|_2=1,\ \|\Gamma\|_0\le s} \Gamma^\top \bLambda \Gamma.
\label{sparse_rayleigh}
\end{align}
Equivalently, it may be characterized through the best sparse rank-one approximation
\begin{align}
(\delta_0,\Gamma_0)
=
\argmin_{\delta\ge 0,\ \|\Gamma\|_2=1,\ \|\Gamma\|_0\le s}
\|\bLambda-\delta \Gamma\Gamma^\top\|_F^2.
\label{sparse_rankone}
\end{align}
At the population level, $\delta_0 = \E\left[\{g'(\bbeta^\top \X)\}^2\right]\|\bbeta\|_2^2$,
and \(\Gamma_0=\bbeta/\|\bbeta\|_2\).

\section{Sample-level estimation}\label{sec:sample_est}
Let \(\{(\x_i,y_i): i=1,\ldots,n\}\) be an independent sample from the distribution of \((\X,Y)\). For each target point \(\x_i\), we estimate a local median slope by solving the weighted \(\ell_1\)-penalized median regression problem
\begin{equation}
(\hat a_i,\hat{\bb}_i)
=
\arg\min_{a,\bb}
\left\{
\sum_{j=1}^n
w_{ij}\left|y_j-a-\bb^\top(\x_j-\x_i)\right|
+
\lambda_n\|\bb\|_1
\right\},
\label{eq:local_regression}
\end{equation}
where \(\lambda_n>0\) controls sparsity in the local slope estimate. In implementation, we use normalized Gaussian radial weights
\begin{equation}
w_{ij}
=
\frac{\exp\left\{-\|\x_j-\x_i\|_2^2/\gamma^2\right\}}
{\sum_{k=1}^n \exp\left\{-\|\x_k-\x_i\|_2^2/\gamma^2\right\}},
\qquad j=1,\ldots,n,
\label{eq:weight_function}
\end{equation}
where \(\gamma>0\) controls the degree of localization. Prior to estimation, each predictor is standardized to have sample mean zero and sample variance one, so that the \(\ell_1\) penalty is applied on a common scale.

Define the matrix of local slope estimates by $\widetilde{\B}_n =(\hat{\bb}_1,\ldots,\hat{\bb}_n)^\top
\in\mathbb R^{n\times p}$. Under the single index conditional median model, the population counterpart of each local slope is asymptotically proportional to the same direction \(\bbeta\). Thus, the rows of \(\widetilde{\B}_n\) are expected to lie approximately in a common one-dimensional linear span up to sampling variability and local estimation error. This motivates extracting the central direction through a sparse rank-one approximation of \(\widetilde{\B}_n\).

Suppose \(\operatorname{rank}(\widetilde{\B}_n)=r\). Its singular value decomposition is
\begin{equation}
\widetilde{\B}_n
=
{\bf U}{\bf D}{\bf V}^\top,
\label{eq:svd}
\end{equation}
where ${\bf U}=(\bm u_1,\ldots,\bm u_r)\in\mathbb R^{n\times r}, \ {\bf V}=(\bm v_1,\ldots,\bm v_r)\in\mathbb R^{p\times r}$, have orthonormal columns, and ${\bf D}=\operatorname{diag}(d_1,\ldots,d_r), \ d_1\ge \cdots \ge d_r>0$. The vectors \(\bm v_k\) are the right singular vectors and may be interpreted as loading directions for \(\widetilde{\B}_n^\top\widetilde{\B}_n\). Since the central subspace is one-dimensional, we focus on its leading rank-one component.

Following the regularized singular value decomposition of \cite{shen2008sparse}, we estimate a sparse rank-one approximation by solving
\begin{equation}
(\hat{\bm u}_1,\hat{\bEta}_1)
=
\arg\min_{\bm u,\bEta}
\left\{
\left\|\widetilde{\B}_n-\bm u\bEta^\top\right\|_F^2
+
\xi\|\bEta\|_1
\right\}, \ \text{subject to } \|\bm u\|_2=1,
\label{eq:spca_objective}
\end{equation}
where \(\xi>0\) is a sparsity tuning parameter. Here, \(\bEta\) plays the role of a sparse loading vector associated with the leading right singular direction. The estimated sparse central direction is then $\hat{\Gamma} = \cfrac{\hat{\bEta}_1}{\|\hat{\bEta}_1\|_2}$,
and the estimated active set is
\begin{equation}
\widehat{\mathcal A}
=
\{j:\hat{\Gamma}_j\neq 0\}.
\label{eq:selected_set}
\end{equation}
The magnitude \(|\hat{\Gamma}_j|\) may be used as a heuristic measure of the contribution of predictor \(X_j\) to the estimated sufficient direction.

\subsection{Implementation issues}
\label{subsec:implementation}

The localization parameter \(\gamma\) governs the effective neighborhood used in each local median regression. Following \cite{LiArtemiouLi2011}, we set
\begin{equation}
\hat{\gamma}
=
\binom{n}{2}^{-1}
\sum_{1\le i<j\le n}\|\x_i-\x_j\|_2,
\label{eq:tau_hat}
\end{equation}
that is, the average pairwise Euclidean distance among the standardized predictors. This value is used in \eqref{eq:weight_function} throughout the estimation procedure.

After the local slopes have been estimated, the sparse rank-one approximation in \eqref{eq:spca_objective} can be computed using, for example, the \texttt{ssvd()} function in the \texttt{irlba} package in $R$, which implements the regularized low-rank approximation of \cite{shen2008sparse}. If the desired number \(s\) of active predictors is known in advance, it can be imposed directly in the sparse SVD step.

Otherwise, we use a practical incremental rule to select \(s\). Let \(\hat d_1(s)\) denote the leading sparse singular value obtained when at most \(s\) nonzero loadings are permitted. Starting from \(s=1\), define
\begin{equation}
\delta_s
=
\frac{\hat d_1(s+1)-\hat d_1(s)}{\hat d_1(s)}.
\label{eq:incremental_s}
\end{equation}
We increase \(s\) until \(\delta_s<\varrho\), where \(\varrho>0\) is a prespecified minimum contribution threshold. This selects the smallest sparsity level for which adding one more nonzero loading yields only a negligible increase in the leading sparse singular component. The complete procedure is summarized in Algorithm 1 in the supplemental appendix.

\section{Theoretical properties}\label{sec:theory}
We establish the theoretical properties under the following regularity conditions in addition to Assumptions \ref{ass_mederror} and \ref{ass_diff}.

\begin{itemize}
\item[(C1)] The support \(\mathrm{Supp}(\X) \subset \R^p \) is compact and convex, and \(\x_0\) is an interior point of \(\mathrm{Supp}(\X)\). The predictors are standardized so that \(\E(X_j)=0\) and \(\E(X_j^2)=1\) for \(j=1,\ldots,p\).

\item[(C2)] The conditional median function \(m(\x)\) has bounded continuous derivatives  on \(\mathrm{Supp}(\X)\) up to order three. Also, \(f_{Y\mid \X}(y\mid \x)\) and its derivative with respect to \(y\) are uniformly bounded in a neighborhood of \(y=m(\x)\).

\item[(C3)] The kernel \(K(\cdot)\) is a bounded, symmetric density with bounded first derivative and finite moments of sufficiently high order.

\item[(C4)] As \(n\to\infty\), \(h\to 0\) and $\cfrac{nh^{p+2}}{C_n}\to\infty$, where $C_n=(\log n)^{C'}$ for some constant \(C'>0\). 

\item[(C5)] The index coefficient \(\bbeta\) is \(s\)-sparse, where \(s=s_n\) may diverge with \(n\), subject to $\cfrac{s\log p}{nh^{p+2}}\to 0$.

\item[(C6)] Let \(\Z=\X-\x_0\) and define $\Sig_h(\x_0)=\E\left\{K_h(\X-\x_0)\Z\Z^\top\right\}$. There exists a constant \(\kappa_0>0\) such that $\delta^\top\Sig_h(\x_0)\delta
\ge \kappa_0\|\delta_{\mathcal A}\|_2^2$ uniformly over \(\x_0\in\mathrm{Supp}(\X)\) and all \(\delta\) satisfying
$\|\delta_{\mathcal A^c}\|_1\le c_0\|\delta_{\mathcal A}\|_1$ for some fixed \(c_0>0\).
\end{itemize}

Conditions (C1)--(C3) are standard local smoothing assumptions, see \cite{fan1996local} and \cite{fan2016multivariate}. Condition (C4) balances localization and effective sample size with the factor \(h^2\) reflecting the estimation of a first derivative rather than a local intercept. Condition (C5) imposes sparsity on the single index direction, and Condition (C6) is a local restricted eigenvalue condition required for Lasso and Dantzig penalties to ensure sufficient curvature of the weighted loss along sparse directions. See \cite{candes2007} and \cite{bickelLasso2009}.

\subsection{Rate of convergence of the local median slope}

For global \(\ell_1\)-penalized quantile regression, \cite{bellonil12011} established convergence rates of order \((s\log p/n)^{1/2}\). For our localized setting, the relevant empirical score for the slope component has effective sample size of order \(nh^p\), while the additional factor \(\Z_i=\X_i-\x_0\) in the local linear slope contributes an extra \(h\)-scaling \citep{fan2016multivariate}.

The choice of \(\lambda_n\) is driven by the sampling variability of the empirical slope score. To determine the appropriate order of \(\lambda_n\), consider the empirical score for the slope component evaluated at the true localized parameter,
\[
S_n=\frac{1}{n}\sum_{i=1}^n K_h(\X_i-\x_0)(\X_i-\x_0)\psi_{1/2}(\epsilon_i).
\]
Because \(K_h(\cdot)=h^{-p}K(\cdot/h)\) and \(\X_i-\x_0\) is of order \(h\) within the local neighborhood,  
\[
\E\left[K_h(\X-\x_0)^2(\X-\x_0)(\X-\x_0)^\top\psi_{1/2}(\epsilon)^2\right]
=O(h^{2-p}).
\]
Thus, each coordinate of \(S_n\) has variance of order \(h^{2-p}/n\), and a concentration argument yields
\[
\|S_n\|_\infty = O_p\left(\sqrt{\frac{h^{2-p}\log p}{n}}\right).
\]
We therefore choose the tuning parameter \(\lambda_n\) to be of this order so that it dominates the random variation of the empirical slope score. 

Let \(\bb_0(\x_0)=\nabla m(\x_0)\),  \(\bb_h(\x_0)\) denote the population local median slope at bandwidth \(h\), and \(\hat{\bb}(\x_0)\) denote the corresponding kernel-weighted local \(\ell_1\)-penalized median slope estimator at \(\x_0\).

\begin{theorem}\label{thm:beta_consistency}
Suppose Assumptions \ref{ass_mederror}--\ref{ass_diff} and Conditions \textup{(C1)}--\textup{(C6)} hold. For an interior point \(\x_0\in\mathrm{Supp}(\X)\), let $\bb_0(\x_0)=\nabla m(\x_0)$. If regularizing parameter $\lambda_n \asymp \bigg(\cfrac{h^{2-p}\log p}{n}\bigg)^{1/2}$, then
\begin{align}
\|\hat{\bb}(\x_0)-\bb_0(\x_0)\|_2 = O_p\left\{
h^2 + \left(\frac{s\log p}{nh^{p+2}}\right)^{1/2}
\right\}.
\label{eq:local_rate}
\end{align}
\end{theorem}

\noindent Theorem \ref{thm:beta_consistency} decomposes the local estimation error into the usual smoothing bias and a stochastic component determined by sparsity, dimensionality, and the effective local sample size. 

Up to multiplicative constants, the squared error corresponding to \eqref{eq:local_rate} is
\begin{align}
\mathrm{MSE}(h)
\asymp
h^4+\frac{s\log p}{nh^{p+2}}.
\label{eq:mse_rate}
\end{align}
Balancing the two terms yields
\begin{align}
h_{\mathrm{opt}}
\asymp (s\log p/n)^{1/(p+6)},
\label{eq:optimal_h}
\end{align}
with a corresponding
\begin{align}
\mathrm{MSE}(h_{\mathrm{opt}})
= O\{\left(s\log p/ n\right)^{4/(p+6)}\}.
\label{eq:optimal_mse}
\end{align}
This rate reflects localization in the ambient predictor space and highlights the role of subsequent low-rank aggregation of the local gradients.

\subsection{Aggregate support recovery}\label{subsec:aggregate_recovery}

We next study recovery of the active set from the aggregated matrix of local gradients. For the observed design points \(\x_1,\ldots,\x_n\), define the population gradient matrix
\begin{align}
\B_{0n} = \bigl(\bb(\x_1),\ldots,\bb(\x_n)\bigr)^\top, \text{ where } \bb(\x_i)=\nabla m(\x_i).
\label{eq:B0}
\end{align}
Under model \ref{eq:median-sim}, $\bb(\x_i)=g'(\bbeta^\top \x_i)\bbeta$. Let \(c_i=g'(\bbeta^\top \x_i)\) and \(\bm c_n=(c_1,\ldots,c_n)^\top\), so that
\begin{align}
\B_{0n} = \bm c_n\bbeta^\top.
\label{eq:rank_one_gradient}
\end{align}
Hence, \(\B_{0n}\) has rank one whenever \(\bm c_n\neq 0\). If \(\|\bbeta\|_2=1\), then $d_{1n}=\|\bm c_n\|_2,
\bu_1=\cfrac{\bm c_n}{\|\bm c_n\|_2}, \bv_1=\bbeta$, and therefore $\B_{0n}=d_{1n}\bu_1\bv_1^\top, \text{ where } \|\bu_1\|_2=\|\bv_1\|_2=1$.

Let
\[
\widetilde{\B}_n = (\hat{\bb}_1,\ldots,\hat{\bb}_n)^\top =\B_{0n} + \bm\varepsilon_n,
\]
where \(\bm\varepsilon_n\) is the matrix of local estimation errors. Although some local gradients may be weak, all rows of \(\B_{0n}\) share the same direction \(\bbeta\). And if $\mu_g=\E\left[\{g'(\bbeta^\top\X)\}^2\right]>0$, then $\cfrac{d_{1n}^2}{n} = \cfrac{1}{n}\displaystyle\sum_{i=1}^n \{g'(\bbeta^\top\x_i)\}^2 \overset{p}{\to}\mu_g$. Thus, \(d_{1n}=O_p(\sqrt n)\). Therefore, for each active coordinate \(j\in\mathcal A\), the aggregate signal grows at rate \(d_{1n}|\beta_j|\), even if \(g'(\bbeta^\top\x_i)\) is small for a substantial fraction of the design points. The following proposition provides a high-level sufficient condition for exact support recovery from the sparse rank-one approximation step.

\begin{proposition}\label{prop:aggregate_selection}
Suppose the conditions of Theorem \ref{thm:beta_consistency} hold and $\B_{0n}=d_{1n}\bu_1\bv_1^\top$, where $\bv_1=\bbeta, \ \|\bbeta\|_2=1$. Let \(\mathcal A=\mathrm{supp}(\bbeta)\) with \(|\mathcal A|=s\), and suppose $\bu_1$ satisfies
$|\hat{\bu}_1^\top \bu_1|\overset{p}{\to} 1$. Assume further that $\|\bm\varepsilon_n^\top \hat{\bu}_1\|_\infty = O_p(r_n)$ for some sequence \(r_n\to 0\), and that the tuning parameter \(\xi_n\) for the sparse loading step is such that
$r_n=o(\xi_n) \text{ and } \xi_n=o\!\left(d_{1n}\min_{j\in\mathcal A}|\beta_j|\right)$. Then $\Pr(\widehat{\mathcal A}=\mathcal A)\to 1$.
\end{proposition}

Proposition \ref{prop:aggregate_selection} shows that exact support recovery can be achieved at the aggregate level without requiring exact support recovery in every local penalized median regression. The local \(\ell_1\) penalty stabilizes high-dimensional estimation of the local median gradients, whereas the sparse rank-one approximation identifies coordinates that persist in the common sufficient direction.

\begin{remark}
This distinction is particularly important when \(g'(\bbeta^\top \x)\) is small on a non-negligible part of the predictor space. In such regions, some individual local fits may fail to recover the active coordinates because the local signal is comparable to the local penalty level. Nevertheless, as long as $\mu_g=\E[\{g'(\bbeta^\top\X)\}^2]>0$, the common direction remains identifiable through aggregation of the local gradients.
\end{remark}

\section{Simulation Study}\label{sec:simulation}
\subsection{Active subset recovery}\label{subsec:subset_recovery}
This subsection examines the variable selection performance of the proposed smOPG procedure under several contamination mechanisms. We compare it with a number of competing \(\ell_1\)-penalized methods, including penalized ordinary least squares (OLS), penalized rank regression (Rank), sparse sliced inverse regression (SIR), sparse minimum average variance estimation (SMAVE), and penalized cumulative quantile regression (CQR). To examine performance across different regimes, we consider both a low-dimensional setting (\(p<n\)) and a high-dimensional setting (\(p \gg n\)). Also, to facilitate comparison with methods that rely on monotonicity or linearity-type conditions, we focus on monotone link functions and elliptically distributed predictors. SMAVE and CQR are omitted from the high-dimensional setting because the former is designed for \(p<n\), whereas the latter becomes computationally prohibitive when \(p\) is very large.

All analyses are conducted in \textsf{R}. We implement OLS- and Rank-Lasso using the \texttt{ncvreg} package, sparse SIR using the \texttt{LassoSIR} package, SMAVE using the \texttt{MAVE} and \texttt{glmnet} packages, and CQR using the \texttt{cqrReg} package. For OLS, Rank, SIR, and SMAVE, tuning is selected by the default cross-validation procedures provided in the corresponding packages, which utilizes 100 values of $\lambda$. For CQR, we used the default value of 1 as suggested by the authors of the \texttt{cqrReg} package.\footnote{The \texttt{cqrReg} package did not have readily available functions for tuning $\lambda$ and the authors noted that the performance of some of the functions such as the ADMM depend on good least squares estimation.} The number of slices in SIR is fixed at 10, and the quantile levels for CQR are fixed at \(\tau=(0.1,0.2,\ldots,0.9)\) in all settings. The proposed smOPG procedure is implemented using the \texttt{conquer} and \texttt{irlba} packages. In the \texttt{conquer} package, $\lambda$ is restricted to the interval \([0,1]\). Therefore, we use choose \(\lambda\) from \((0.01,0.05,0.1,0.25,0.5,0.9)\) based on cross-validation. In fact, preliminary experiments indicated that values of \(\lambda>0.2\) typically shrink nearly all local slope estimates to zero, which in turn yields an empty or nearly empty selected support after aggregation. We maintained the given set to ensure fairness and allow for worse case scenarios.

The data are generated as follows. We consider \((n,p)\in\{(50,10),(100,1000)\}\), and fix the contamination proportion at \(5\%\). Thus, \(n_0=0.95n\) observations are generated from the baseline model, while the remaining \(n_1=n-n_0\) observations are contaminated. Let
\[
\bbeta= (\bbeta_{\mathcal A}^\top, \ \bbeta_{\mathcal A^c}^\top)^\top,
\quad
\bbeta_{\mathcal A}=(-2,-1.5,0.8,1.5,3)^\top,
\qquad
\bbeta_{\mathcal A^c}=\bm 0\in\mathbb R^{p-5}.
\]
The coordinates of \(\bbeta\) are randomly permuted so that each predictor is equally likely to belong to the active set.

Let \(\Sig\in\mathbb R^{p\times p}\) have entries \((\Sig)_{kl}=0.75^{|k-l|}\), \(k,l=1,\ldots,p\). For \(i=1,\ldots,n_0\), the uncontaminated predictors are generated as \(\x_i\sim N_p(\bm 0,\Sig)\), and the response is generated from one of the following models:
\begin{enumerate}
\item[I.] \(y_i = 0.5+\bbeta^\top\x_i+0.8\epsilon_i\),
\item[II.] \(y_i = \exp(0.05\bbeta^\top\x_i+1)+(0.2+0.01x_{1i})\epsilon_i\),
\item[III.] \(y_i = 0.5(\bbeta^\top\x_i)^3+0.8\epsilon_i\),
\end{enumerate}
where \(\epsilon_i\sim N(0,1)\). For the contaminated observations \(i=n_0+1,\ldots,n\), the response is generated from the same model, but the predictors and/or errors are contaminated according to one of the following mechanisms:
\begin{itemize}
\item \textbf{predictor contamination only:} \(\x_i\) follows a \(p\)-variate \(t\) distribution with 10 degrees of freedom, location vector \(\bm\mu=(5, \ldots,5)\in\R^p\), and scatter matrix \(\Sig\), while \(\epsilon_i\sim N(0,1)\);
\item \textbf{response contamination only:} \(\x_i\sim N_p(\bm 0,\Sig)\), while \(\epsilon_i\sim \mathrm{Cauchy}(0,1)\);
\item \textbf{combined predictor and response contamination:} \(\x_i\) follows the same multivariate \(t\) distribution as above, and \(\epsilon_i\sim \mathrm{Cauchy}(0,1)\).
\end{itemize}

The strong correlation among nearby predictors allows us to assess the behavior of each method under correlated covariates, especially with respect to false positives and support instability. Models I and III have monotone regression links and homoscedastic errors, so methods that rely on monotonicity are expected to be relatively more competitive in these settings. Model II is heteroscedastic, and is therefore expected to be more challenging, particularly for methods whose performance depends strongly on mean structure or response ordering.

For all models and settings, we evaluate performance using the average selected model size
$\hat s=\bigl|\{j\in\{1,\ldots,p\}:\hat b_j\neq 0\}\bigr|$,
the true positive rate (TPR) and false positive rate (FPR) respectively given by
\[
\mathrm{TPR} = \cfrac{|\{j:\hat b_j\neq 0,\ \beta_j\neq 0\}|}{s}, \qquad \mathrm{FPR}=\cfrac{|\{j:\hat b_j\neq 0,\ \beta_j=0\}|}{p-s}.
\]
Each setting is replicated 100 times.

\begin{table}[t!]
\caption{Mean (standard deviation) of the active set size $\hat s$, true positive rate (TPR), and false positive rate (FPR) for $(n,p) = (50, 10)$ based on 100 random samples. }
\label{tab:sims_low_dim}
\resizebox{\textwidth}{!}{
\begin{tabular}{cl|ccc|ccc|ccc}
\hline
\multirow{2}{*}{Model} & \multicolumn{1}{c|}{\multirow{2}{*}{Method}} & \multicolumn{3}{c|}{Predictor contamination only}  & \multicolumn{3}{c|}{Response contamination only} & \multicolumn{3}{c}{Combined contamination} \\ \cline{3-11} 
& \multicolumn{1}{c|}{}  & $\hat s$ & TPR  & FPR  & $\hat s$ & TPR & FPR & $\hat s$ & TPR & FPR  \\ \hline
\multirow{6}{*}{I} & CQR & 10.00 (0.00) & 1.00 (0.00) & 1.00 (0.00) & 10.00 (0.00) & 1.00 (0.00) & 1.00 (0.00) & 10.00 (0.00)   & 1.00 (0.00)   & 1.00 (0.00)   \\
& OLS & 8.15 (1.36)  & 1.00 (0.00) & 0.63 (0.27) & 7.83 (2.04)  & 0.95 (0.18) & 0.61 (0.3)  & 7.83 (2.04)    & 0.95 (0.18)   & 0.61 (0.3)    \\
& RANK & 7.95 (1.47)  & 0.98 (0.07) & 0.61 (0.28) & 8.08 (1.44)  & 0.98 (0.07) & 0.64 (0.27) & 8.08 (1.44)    & 0.98 (0.07)   & 0.64 (0.27)   \\
& SIR & 6.47 (2.3)   & 0.76 (0.25) & 0.53 (0.29) & 6.26 (2.22)  & 0.74 (0.24) & 0.51 (0.28) & 6.26 (2.22)    & 0.74 (0.24)   & 0.51 (0.28)   \\
& SMAVE & 9.74 (0.48)  & 1.00 (0.00) & 0.95 (0.1)  & 9.73 (0.67)  & 1.00 (0.02) & 0.95 (0.13) & 9.73 (0.67)    & 1.00 (0.02)   & 0.95 (0.13)   \\
& smOPG  & {\bf 4.93} (1.04)  & 0.94 (0.16) & { 0.05} (0.11) & {\bf 4.94} (1.07)  & 0.93 (0.16) & { 0.06} (0.11) & {\bf 4.94} (1.07)    & 0.93 (0.16)   & {0.06} (0.11)   \\ \hline
\multirow{6}{*}{II}    & CQR & 10.00 (0.00) & 1.00 (0.00) & 1.00 (0.00) & 10.00 (0.00) & 1.00 (0.00) & 1.00 (0.00) & 10.00 (0.00)   & 1.00 (0.00)   & 1.00 (0.00)   \\
& OLS & 6.3 (2.48)   & 0.75 (0.24) & 0.52 (0.31) & {\bf 4.83} (2.95)  & 0.57 (0.33) & 0.39 (0.32) & {\bf 4.83} (2.95)    & 0.57 (0.33)   & 0.39 (0.32)   \\
& RANK & 5.53 (2.4)   & 0.69 (0.24) & 0.42 (0.3)  & {\bf 5.17} (2.51)  & 0.64 (0.28) & 0.39 (0.28) & {\bf 5.17} (2.51)    & 0.64 (0.28)   & 0.39 (0.28)   \\
& SIR & 4.47 (1.98)  & 0.52 (0.24) & 0.38 (0.27) & 4.13 (2.19)  & 0.46 (0.26) & 0.37 (0.27) & 4.13 (2.19)    & 0.46 (0.26)   & 0.37 (0.27)   \\
& SMAVE & 9.64 (0.65)  & 0.97 (0.07) & 0.95 (0.1)  & 9.7 (0.48)   & 0.98 (0.07) & 0.96 (0.08) & 9.7 (0.48)     & 0.98 (0.07)   & 0.96 (0.08)   \\
& smOPG & {\bf 4.56} (1.73)  & 0.62 (0.23) & 0.29 (0.21) & 4.24 (1.69)  & 0.59 (0.21) & 0.26 (0.22) & 4.24 (1.69)    & 0.59 (0.21)   & 0.26 (0.22)   \\ \hline
\multirow{6}{*}{III}   & CQR  & 10.00 (0.00) & 1.00 (0.00) & 1.00 (0.00) & 10.00 (0.00) & 1.00 (0.00) & 1.00 (0.00) & 10.00 (0.00)   & 1.00 (0.00)   & 1.00 (0.00)   \\
& OLS & 6.54 (2.65)  & 0.74 (0.28) & 0.56 (0.31) & 6.46 (2.70)  & 0.74 (0.28) & 0.55 (0.33) & 6.46 (2.70)    & 0.74 (0.28)   & 0.55 (0.33)   \\
& RANK & 7.63 (1.47)  & 0.98 (0.08) & 0.55 (0.27) & 7.67 (1.50)  & 0.98 (0.08) & 0.56 (0.28) & 7.67 (1.50)    & 0.98 (0.08)   & 0.56 (0.28)   \\
& SIR & 6.58 (2.26)  & 0.78 (0.26) & 0.54 (0.27) & 6.49 (2.27)  & 0.76 (0.26) & 0.53 (0.26) & 6.49 (2.27)    & 0.76 (0.26)   & 0.53 (0.26)   \\
& SMAVE & 9.45 (0.75)  & 1.00 (0.02) & 0.89 (0.15) & 9.47 (0.72)  & 1.00 (0.02) & 0.9 (0.14)  & 9.47 (0.72)    & 1.00 (0.02)   & 0.9 (0.14)    \\
& smOPG & {\bf 4.84} (1.94)  & 0.74 (0.23) & 0.23 (0.27) & {\bf 4.65} (2.07)  & 0.72 (0.24) & 0.21 (0.27) & {\bf 4.65} (2.07)  & 0.72 (0.24)   & 0.21 (0.27)   \\ \hline
\end{tabular}
}
\end{table}

Tables~\ref{tab:sims_low_dim}--\ref{tab:sims_high_dim} and Figures~\ref{fig:sims_low_dim}--\ref{fig:sims_high_dim}, summarize the simulation results under the different contamination mechanisms. The ROC plots in Figure~\ref{fig:sims_low_dim} illustrate the tradeoff between true positive and false positive rates in the low-dimensional setting. In the ultra high-dimensional setting, because $|p-s|$ is very large, the FPR appears low for all methods despite the high false selections in some methods, making the ROC plot less visually informative. Therefore, we report the ratio of the estimated TPR to the active set size $\hat s$, which is proportional to the precision, i.e, the proportion of selected predictors that are truly active.

\begin{figure}[htb!]
    \centering
    \includegraphics[width=.9\linewidth]{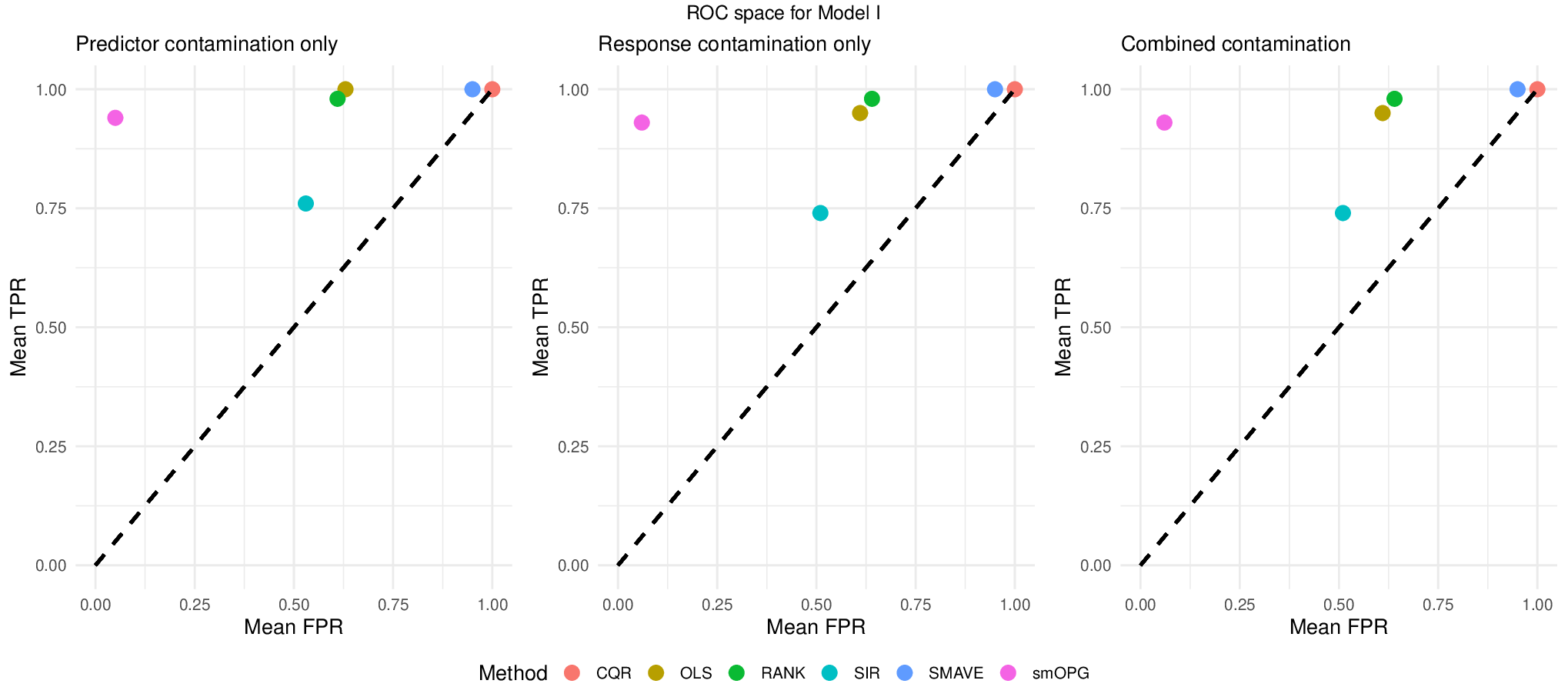}
    \includegraphics[width=.9\linewidth]{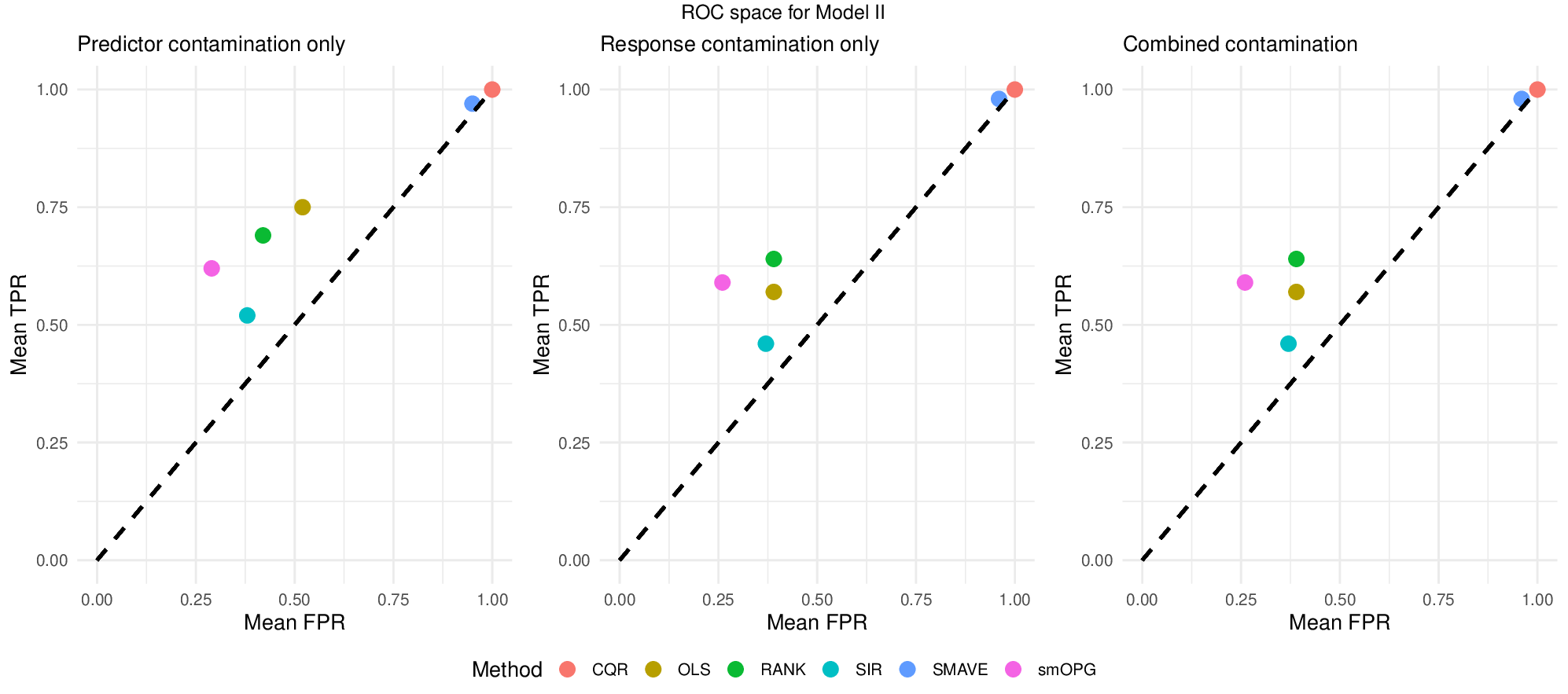}
    \includegraphics[width=.9\linewidth]{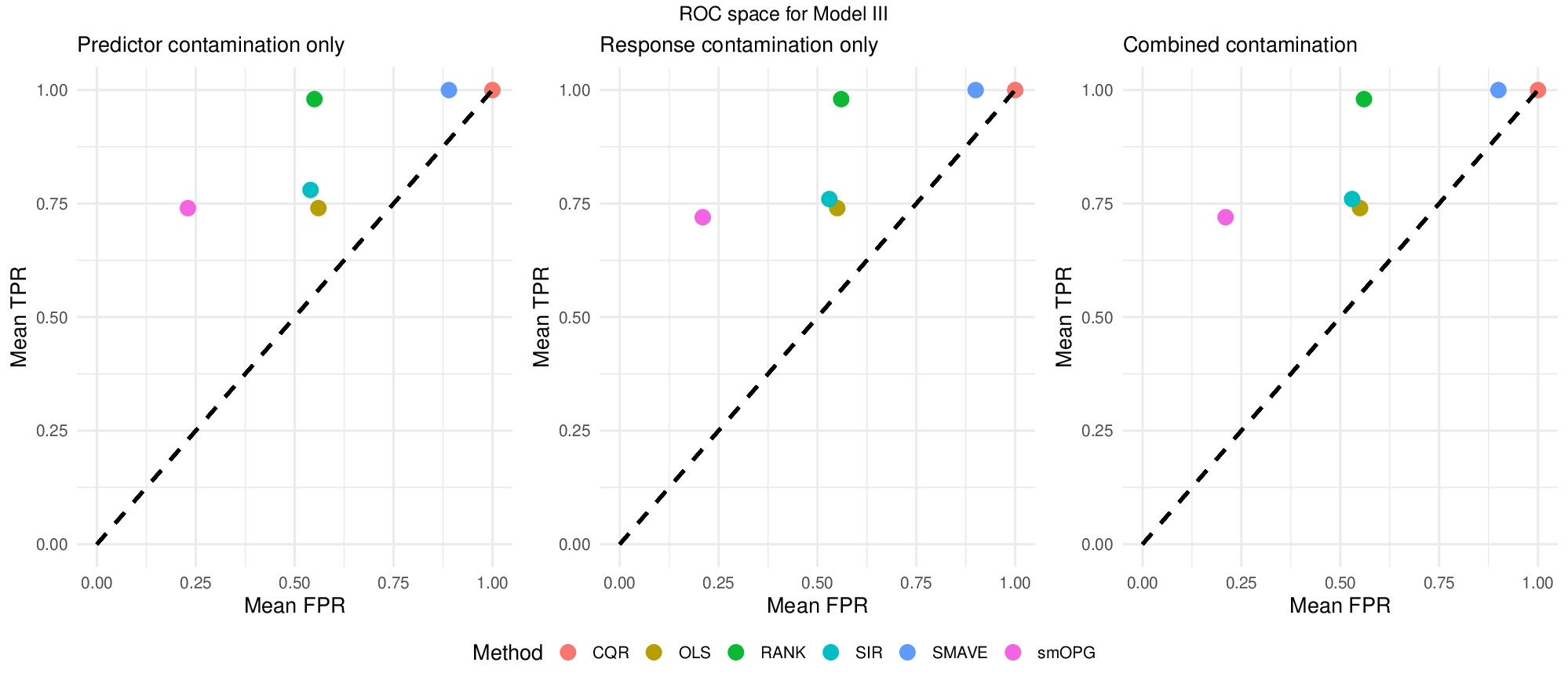}
    \vspace{-0.2in}
    \caption{Average TPR vs FPR based on on $(n,p)=(50, 10)$ with $5\%$ contamination based on 100 random samples }
    \label{fig:sims_low_dim}
\end{figure}

The simulation results reveals several patterns. In the low-dimensional setting, most methods perform reasonably well under Model I, where the regression structure is linear and the contamination is comparatively mild. As the settings become more challenging, the performance of all methods deteriorate. Comparatively, the proposed smOPG performed better than its competitors across settings with the differences more pronounced in Models I and III. In particular, under response only contamination and combined contamination, the proposed smOPG procedure tends to maintain a stronger balance between true positive recovery and false positive control compared to the competing techniques. This advantage is especially visible in Model II, where heteroscedasticity complicates direct estimation of the mean structure and response ordering. Model III remains monotone but introduces nonlinearity, and here too smOPG remains competitive while avoiding excessive model size. The contrast is sharper in the high-dimensional setting, where all methods deteriorate to some extent, but the proposed procedure remains comparatively stable under various contamination mechanisms. 

In the ROC plots in Figure~\ref{fig:sims_low_dim}, methods closer to the upper-left corner achieve stronger recovery of active predictors while limiting spurious selections. Thus, across settings, smOPG remains the most competitive. CQR and SMAVE lie very close to the dotted line, suggesting that their selection is not different from a random guess based on a fair coin flip. The proposed smOPG also shows the highest efficiency per additional variable included in the active set compared to the rest as displayed in Figure~\ref{fig:sims_high_dim}. Overall, the simulation results support the use of local median smoothing together with sparse rank-one aggregation for robust active subset recovery in contaminated single index models.


\begin{table}[t!]
\caption{Mean (standard deviation) of the active set size $\hat s$, true positive rate (TPR), and false positive rate (FPR) for $(n,p) = (100, 1000)$ based on 100 random samples. }
\label{tab:sims_high_dim}
\resizebox{\textwidth}{!}{
\begin{tabular}{cl|ccc|ccc|ccc}
\hline
\multirow{2}{*}{Model} & \multicolumn{1}{c|}{\multirow{2}{*}{Method}} & \multicolumn{3}{c|}{Predictor contamination only}        & \multicolumn{3}{c|}{Response contamination only}  & \multicolumn{3}{c}{Combined contaminationl} \\ \cline{3-11} 
& \multicolumn{1}{c|}{}  & $\hat s$  & TPR & FPR  & $\hat s$ & TPR  & FPR  & $\hat s$ & TPR  & FPR  \\ \hline
\multirow{4}{*}{I} & OLS & 14.50 (4.49)  & 1.00 (0.02) & 0.01 (0.00) & 20.00 (11.20) & 0.91 (0.25) & 0.02 (0.01) & 20.00 (11.20)   & 0.91 (0.25)   & 0.02 (0.01)  \\
& RANK  & 39.00 (18.6)  & 0.99 (0.04) & 0.03 (0.02) & 38.8 (17.10)  & 0.99 (0.05) & 0.03 (0.02) & 38.80 (17.10)   & 0.99 (0.05)   & 0.03 (0.02)  \\
& SIR   & 9.47 (4.07)   & 0.29 (0.13) & 0.01 (0.00) & 10.20 (3.97)  & 0.29 (0.13) & 0.01 (0.00) & 10.20 (3.97)    & 0.29 (0.13)   & 0.01 (0.00)  \\
& smOPG & {\bf 4.31} (0.94)   & 0.82 (0.17) & 0.00 (0.00) & {\bf 4.16} (1.20)   & 0.78 (0.21) & 0.00 (0.00) & {\bf 4.16} (1.20)     & 0.78 (0.21)   & 0.00 (0.00)  \\ \hline
\multirow{4}{*}{II}  & OLS  & 27.4 (15.70)  & 0.67 (0.22) & 0.02 (0.02) & 17.50 (17.60) & 0.40 (0.35) & 0.02 (0.02) & 17.50 (17.60)   & 0.40 (0.35)   & 0.02 (0.02)  \\
& RANK & 25.4 (16.30)  & 0.64 (0.2)  & 0.02 (0.02) & 29.30 (23.50) & 0.60 (0.20) & 0.03 (0.02) & 29.30 (23.50)   & 0.60 (0.20)   & 0.03 (0.02)  \\
& SIR   & 10.2 (3.50)   & 0.21 (0.13) & 0.01 (0.00) & 10.20 (3.73)  & 0.16 (0.15) & 0.01 (0.00) & 10.20 (3.73)    & 0.16 (0.15)   & 0.01 (0.00)  \\
& smOPG & {\bf 2.99} (2.89)   & 0.32 (0.17) & 0.00 (0.00) & {\bf 4.58} (4.72)   & 0.33 (0.19) & 0.00 (0.00) & {\bf 4.58} (4.72)     & 0.33 (0.19)   & 0.00 (0.00)  \\ \hline
\multirow{4}{*}{III}   & OLS  & 32.00 (17.90) & 0.65 (0.26) & 0.03 (0.02) & 31.40 (16.90) & 0.65 (0.26) & 0.03 (0.02) & 31.40 (16.90)   & 0.65 (0.26)   & 0.03 (0.02)  \\
& RANK & 39.2 (17.60)  & 1.00 (0.02) & 0.03 (0.02) & 40.00 (18.4)  & 1.00 (0.02) & 0.04 (0.02) & 40.00 (18.40)   & 1.00 (0.02)   & 0.04 (0.02)  \\
& SIR  & 10.30 (4.26)  & 0.32 (0.13) & 0.01 (0.00) & 10.30 (4.45)  & 0.33 (0.13) & 0.01 (0.00) & 10.30 (4.45)    & 0.33 (0.13)   & 0.01 (0.00)  \\
& smOPG & {\bf 4.56} (2.39)   & 0.54 (0.24) & 0.00 (0.00) & {\bf 4.41} (1.93)   & 0.55 (0.24) & 0.00 (0.00) & {\bf 4.41} (1.93)     & 0.55 (0.24)   & 0.00 (0.00)  \\ \hline
\end{tabular}
}
\end{table}

\begin{figure}[htb!]
    \centering
    \includegraphics[width=1.0\linewidth]{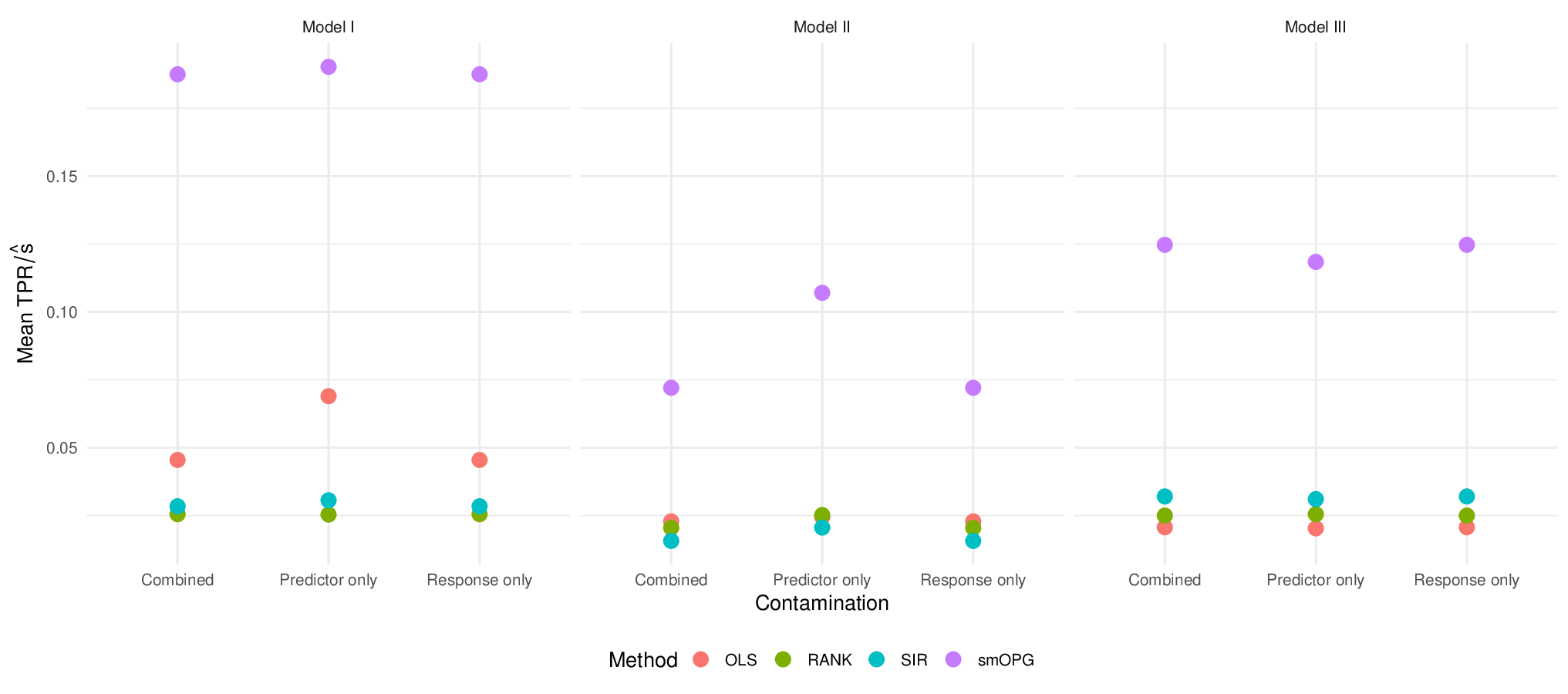}
    \caption{The ratio of the average TPR to the active set size $\hat s$ based on $(n,p)=(100, 1000)$ with $5\%$ contamination based on 100 random samples.}
    \label{fig:sims_high_dim}
\end{figure}

\subsection{Effect of tuning parameters}

The proposed smOPG procedure involves tuning parameters at two stages of estimation: the local \(\ell_1\)-penalized median regressions and the subsequent sparse rank-one approximation of the outer-product of gradients. We therefore examine the sensitivity of the method to the local regularization parameter \(\lambda\) and to the contribution threshold \(\varrho\) used for selecting the final support size.

We first study the effect of \(\lambda\) on active-set recovery while holding the contribution threshold fixed at \(\varrho=0.5\%\). In our implementation using the \texttt{conquer} package, the local regularization parameter is restricted to the interval \([0,1]\). Preliminary experiments indicated that values of \(\lambda>0.2\) typically shrink nearly all local slope estimates to zero, which in turn yields an empty or nearly empty selected support after aggregation. We therefore restrict attention to \(\lambda\le 0.2\). All other simulation settings are the same as in Section \ref{subsec:subset_recovery}, except that we use \((n,p)=(100,50)\) in order to examine behavior in a moderate-dimensional regime.

\begin{table}[htb!]
\caption{Mean (standard deviation) of the active set size $\hat s$, true positive rate (TPR), and false positive rate (FPR) for $(n,p) = (100, 1000)$ based on 100 random samples. }
\label{tab:sims_lambda}
\resizebox{\textwidth}{!}{
\begin{tabular}{l|lll|lll|lll}
\hline
& \multicolumn{3}{c|}{Model I}  & \multicolumn{3}{c|}{Model II}  & \multicolumn{3}{c}{Model III}  \\ \hline
$\lambda$ & \multicolumn{1}{c}{$\hat s$} & \multicolumn{1}{c}{TPR} & \multicolumn{1}{c|}{FPR} & \multicolumn{1}{c}{$\hat s$} & \multicolumn{1}{c}{TPR} & \multicolumn{1}{c|}{FPR} & \multicolumn{1}{c}{$\hat s$} & \multicolumn{1}{c}{TPR} & \multicolumn{1}{c}{FPR} \\ \hline
0.001  & {\bf 5.29} (1.67) & 0.88 (0.22) & 0.02 (0.03)  & 6.90 (3.38) & {\bf 0.66} (0.28)  & 0.08 (0.06)  & 8.43 (4.50) & { 0.60} (0.21) & 0.12 (0.09)  \\
0.01   & {\bf 4.82} (1.01) & {\bf 0.91} (0.20) & 0.01 (0.02) & {\bf 4.80} (2.44) & 0.56 (0.23) & 0.04 (0.04)  & 7.69 (3.90)  & 0.57 (0.22) & 0.11 (0.08) \\
0.10   & 3.78 (1.30) & 0.70 (0.23)  & 0.01 (0.01) & 2.10 (0.99) & 0.28 (0.10) & 0.02 (0.01) & {\bf 3.92} (1.47) & {\bf 0.63} (0.22) & 0.02 (0.02)  \\
0.15   & 2.03 (1.04) & 0.36 (0.18) & 0.01 (0.01) & 1.40 (0.70) & 0.24 (0.08) & 0.00 (0.01) & 2.46 (1.06)  & 0.39 (0.16)  & 0.01 (0.02) \\
0.20   & 1.34 (0.61) & 0.25 (0.10)  & 0.00 (0.01)  & 1.00 (0.00) & 0.16 (0.08) & 0.00 (0.01)  & 1.53 (0.63) & 0.27 (0.10)  & 0.00 (0.01) \\ \hline
\end{tabular}
}
\end{table}

Table~\ref{tab:sims_lambda} suggests that the proposed procedure performs best for relatively small values of \(\lambda\), particularly in the range \([0.001,0.01]\). This indicates that overly aggressive local penalization may remove useful signal from the local median gradients before the global aggregation step is applied. By contrast, smaller values of \(\lambda\) preserve more local directional information, which appears to improve subsequent support recovery while still benefiting from the robustness of the median-based local loss.

We next investigate the effect of the contribution threshold \(\varrho\), while keeping all other settings unchanged. Recall that \(\varrho\) determines the minimum relative increase in the leading sparse singular component required to admit an additional active variable. The results in Table~\ref{tab:sims_delta} show that smaller thresholds, especially those at or below \(1\%\), tend to produce larger selected models and higher true positive rates. Among the values considered, \(\varrho=0.5\%\) appears to provide the best overall balance between recovering active variables and controlling false positives, with an average selected model size close to the true support size.

\begin{table}[htb!]
\caption{Mean (standard deviation) of the active set size $\hat s$, true positive rate (TPR), and false positive rate (FPR) for $(n,p) = (100, 1000)$ based on 100 random samples. }
\label{tab:sims_delta}
\resizebox{\textwidth}{!}{
\begin{tabular}{c|lll|lll|lll}
\hline
\multicolumn{1}{l|}{} & \multicolumn{3}{c|}{Model I} & \multicolumn{3}{c|}{Model II} & \multicolumn{3}{c}{Model III} \\ \hline
\multicolumn{1}{l|}{$\varrho$} & \multicolumn{1}{c}{$\hat s$} & \multicolumn{1}{c}{TPR} & \multicolumn{1}{c|}{FPR} & \multicolumn{1}{c}{$\hat s$} & \multicolumn{1}{c}{TPR} & \multicolumn{1}{c|}{FPR} & \multicolumn{1}{c}{$\hat s$} & \multicolumn{1}{c}{TPR} & \multicolumn{1}{c}{FPR} \\ \hline
0.1\%        & 6.64 (3.08)              & {\bf 0.98} (0.13)             & 0.04 (0.07)              & 11.30 (3.91)             & {\bf 0.72} (0.19)             & 0.17 (0.08)              & 16.6 (8.65)              & {\bf 0.70} (0.21)             & 0.29 (0.18)             \\
0.5\%  & {\bf 4.88} (1.05)              & 0.91 (0.21)             & 0.01 (0.02)              & 6.46 (3.02)              & 0.6 (0.20)              & 0.08 (0.06)              & 7.75 (3.94)              & 0.57 (0.22)             & 0.11 (0.08)             \\
1\%  & {\bf 4.55} (1.02)              & 0.88 (0.22)             & 0.00 (0.02)              & {\bf 4.59} (2.13)              & 0.52 (0.20)             & 0.04 (0.04)              & {\bf 5.30} (2.90)                & 0.51 (0.22)             & 0.06 (0.05)             \\
5\%                        & 2.77 (1.14)              & 0.55 (0.24)             & 0.00 (0.01)              & 2.32 (1.00)              & 0.39 (0.16)             & 0.01 (0.01)              & 2.63 (1.24)              & 0.36 (0.19)             & 0.02 (0.02)             \\
10\%                       & 1.78 (0.63)              & 0.35 (0.13)             & 0.00 (0.01)              & 1.93 (0.76)              & 0.35 (0.15)             & 0.00 (0.01)              & 2.04 (0.92)              & 0.31 (0.16)             & 0.01 (0.02)             \\ \hline
\end{tabular}
}
\end{table}

Overall, these experiments suggest that the proposed procedure is relatively stable when the local regularization parameter is chosen to be small and the aggregation threshold is set at a moderate level. In the simulation studies reported below, we therefore use values of \(\lambda\) in the range \([0.001,0.01]\) together with \(\varrho=0.5\%\), unless otherwise noted. Additional simulation studies involving symmetric-link and low signal-to-noise settings are deferred to the supplemental appendix.

\section{Real Application}\label{sec:realdata}
In this section, we consider two real data applications. For each dataset, we apply the same variable selection procedures considered in the simulation study, where applicable. To compare the quality of the selected active set on a common basis, we refit an ordinary least squares (OLS) model using only the variables selected by each method and evaluate prediction from this reduced model. Each dataset is repeatedly partitioned into training and test sets so that both selection stability and out-of-sample predictive performance can be assessed.

For each method, we report the proportion of times each variable is selected across repeated data splits as a measure of selection stability. We also assess the selected subset \(\X_{\mathcal A}\) using the distance correlation (dCor) measure of \citet{dcorSzekely2007}, which measures the strength of joint dependence between \(Y\) and \(\X_{\mathcal A}\). Unlike Pearson correlation, distance correlation detects both linear and nonlinear dependence between random vectors of arbitrary dimension, and it equals zero if and only if the two vectors are independent. To complement it, we report the root mean squared prediction error (RMSE) obtained from the post-selection OLS refit on the test data. This provides a common prediction benchmark across methods. 

\subsection{Application to Pollution Data}\label{subsec:pollution_design}
We analyze the US air pollution dataset collected by \cite{mcdonald1960air} and previously studied by \citet{gencc2025weighted} and \citet{yuzbacsi2020shrinkage}.
The data comprise $n=60$ US cities in 1960 with response $y$ as the total age--adjusted mortality rate (per $100{,}000$) and $p=15$ predictors
summarizing meteorological, demographic, socioeconomic, and pollution-related factors. Standard collinearity diagnostics indicate substantial linear dependence among predictors with variance inflation factors for the hydrocarbon (HC) and nitric oxide (NOX) approximately $98.6$ and $105.0$, respectively.

\begin{figure}[htb!]
    \centering
    \includegraphics[width=0.99\linewidth]{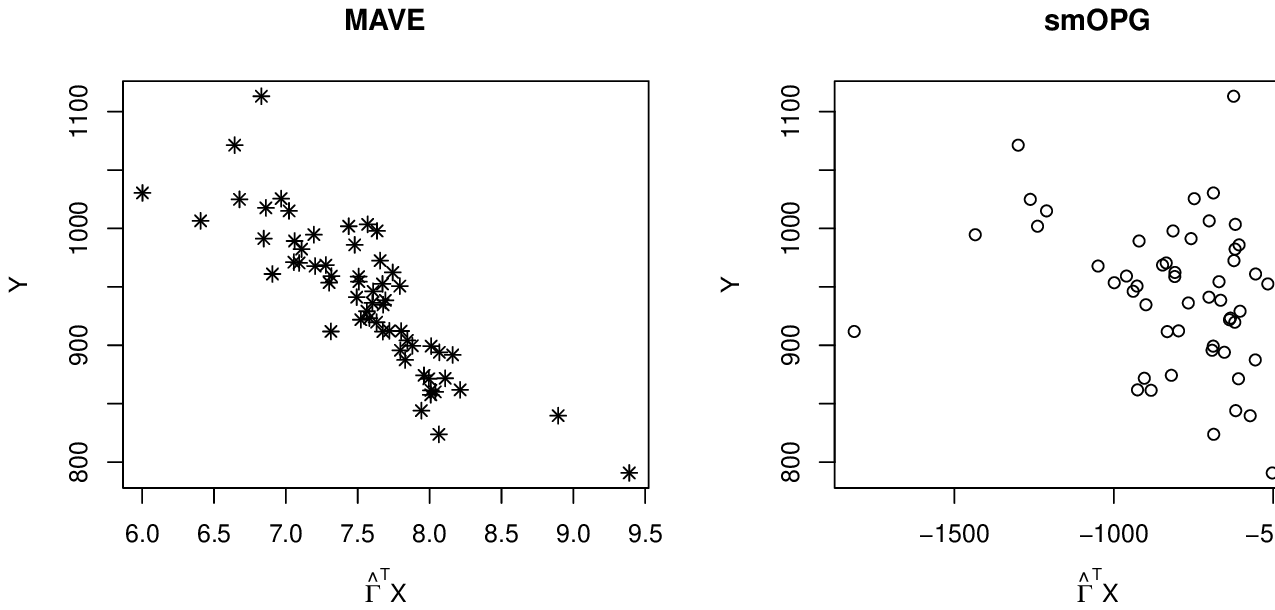}
    \caption{Estimated sufficient summary plot based on classical MAVE and smOPG}
    \label{fig:essp_pollution}
\end{figure}

We begin by examining the estimated sufficient summary plot of \(Y\) versus \(\hat\Gamma^\top\X\) based on the full sample, using both classical minimum average variance estimation (MAVE; \cite{xia2002adaptive}) and the proposed smOPG, in order to visualize the underlying regression structure. As shown in Figure~\ref{fig:essp_pollution}, the relationship appears to be approximately linear, although some observations appear to depart noticeably from the main trend. These extreme observations appear less separated from the bulk of the data for the MAVE summary than under smOPG. This is consistent with the fact that MAVE is mean-based and therefore more influenced by large responses, whereas the median-based smOPG is less sensitive to extreme observations.

The estimated basis vectors also suggest markedly different structures. The estimated MAVE direction is
\[
(-0.02, 0.02, 0.02, 0.08, 0.99, 0.02, 0.02, 0.00, -0.04, 0.01, -0.01, 0.01, -0.01, 0.00, 0.01)^\top,
\]
whereas for smOPG it is
\[
(-0.29, 0.00, 0.00, 0.00, 0.00, 0.00, 0.00, -0.18, -0.84, 0.00, 0.00, 0.00, 0.00, -0.42, 0.00)^\top.
\]
These estimates suggest that only a small subset of predictors contributes strongly to the estimated index. We therefore proceed to a more systematic variable selection analysis.

To assess selection stability and out-of-sample performance, the data are repeatedly partitioned into a training set of size 40 and a test set of size 20. All tuning and model fitting are carried out exclusively on the training set, and the predictors are centered and scaled prior to estimation. For each method, Table~\ref{tbl:pollution} reports the average selected model size, the post-selection distance correlation between the response and the selected predictor subset, and the prediction error from the post-selection OLS refit. Unless otherwise noted, the implementation and tuning settings follow those used in the simulation study in Section~\ref{subsec:subset_recovery}.

\begin{table}[htb!]
\centering
\caption{Train selection proportion in 100 random samples -- proportions $\geq  50\%$ are emboldened. dCor and RMSE are based on test sample.}
\label{tbl:pollution}
\resizebox{0.9\textwidth}{!}{
\begin{tabular}{lcccccc}
\hline
Predictor                         & CQR           & OLS           & SIR           & SMAVE         & RANK           & smOPG         \\ \hline
Avg annual precipitation          & \textbf{1.00} & \textbf{0.97} & 0.42          & \textbf{0.99} & \textbf{0.95}  & \textbf{0.80} \\
Avg Jan temperature               & \textbf{1.00} & \textbf{0.74} & 0.30          & \textbf{0.98} & \textbf{0.81}  & 0.32          \\
Avg Jul temperature               & \textbf{1.00} & \textbf{0.52} & 0.17          & \textbf{0.99} & \textbf{0.53}  & 0.10          \\
\% Population 65+                 & \textbf{1.00} & 0.11          & 0.11          & \textbf{0.93} & 0.16           & 0.11          \\
Avg household size                & \textbf{1.00} & 0.23          & 0.31          & \textbf{0.99} & 0.29           & 0.07          \\
Median school years               & \textbf{1.00} & \textbf{0.86} & 0.45          & \textbf{0.95} & \textbf{0.74}  & 0.41          \\
\% Housing units                  & \textbf{1.00} & \textbf{0.55} & 0.39          & \textbf{0.95} & \textbf{0.63}  & 0.17          \\
Urban population per sq. mile     & \textbf{1.00} & \textbf{0.82} & \textbf{0.99} & \textbf{0.97} & \textbf{0.77}  & 0.47          \\
\% non-White population           & \textbf{1.00} & \textbf{1.00} & \textbf{0.59} & \textbf{1.00} & \textbf{1.00}  & \textbf{0.99} \\
\% White collar occupations       & \textbf{1.00} & \textbf{0.52} & 0.28          & \textbf{0.94} & \textbf{0.52}  & 0.13          \\
\% family income \textless \$3000 & \textbf{1.00} & 0.19          & 0.16          & \textbf{0.95} & 0.18           & 0.03          \\
Rel. HC pollution potential       & \textbf{1.00} & 0.13          & 0.11          & \textbf{0.97} & 0.32           & 0.39          \\
Rel. NOX pollution potential      & \textbf{1.00} & 0.11          & 0.19          & \textbf{0.98} & 0.05           & 0.42          \\
Rel. SO\textsubscript{2} pollution potential      & \textbf{1.00} & \textbf{0.98} & \textbf{0.51} & \textbf{0.91} & \textbf{0.99}  & \textbf{0.80} \\
Annual avg \% rel. humidity       & \textbf{1.00} & 0.34          & 0.23          & \textbf{0.97} & 0.25           & 0.03          \\ \hline
$\hat s$   & 14.47         & 8.07          & 5.21          & 15.00         & 8.19           & 5.24          \\
$dCor(Y, \X_{\hat{\mathbf A}})$                              & 0.43          & 0.44          & 0.42          & 0.43          & 0.44           & \textbf{0.47} \\
RMSE                              & 53.84         & 45.94         & 52.91         & 53.37         & \textbf{44.22} & 46.57         \\ \hline
\end{tabular}
}
\end{table}

Table~\ref{tbl:pollution} summarizes the selection proportions, together with the average test RMSE and the average distance correlation between \(Y\) and \(\X_{\widehat{\mathcal A}}\). As in the simulation study, CQR and SMAVE tend to select nearly all predictors. Penalized OLS and Rank yield very similar results, selecting roughly 66\% of the variables on average. By contrast, smOPG and sparse SIR are substantially more parsimonious, each producing an average model size of about 3 predictors. Although sparse SIR is similarly sparse, the subset selected by smOPG includes predictors representing weather, socio-demographic, and pollution-related domains. In this sense, the smOPG model captures a broader range of relevant predictor types. It also exhibits stronger joint dependence with the response, as reflected by the distance correlation measure, while remaining competitive in terms of prediction error.

Across methods, average annual rainfall, socio-demographic variables related to urban structure (such as the percentage of non-White population), and pollution variables associated with SO\textsubscript{2} are selected with high frequency. Overall, this application suggests that smOPG provides a competitive and parsimonious alternative to both classical penalization methods and existing sparse dimension reduction procedures, with stable selection patterns, strong dependence with the response, and favorable out-of-sample prediction performance.

\subsection{Application to Riboflavin Data}\label{subsec:riboflavin}
We next consider the \emph{riboflavin} dataset, a high-dimensional genomic study on \textit{Bacillus subtilis} that has become a standard benchmark for sparse regression in the $p \gg n$ regime. The data was introduced in \citet{buhlmann2014high} and are distributed in \texttt{R} via the \texttt{hdi} package. The goal is to identify genes whose expression levels are associated with increased riboflavin
(\textit{vitamin B2}) production, with the longer-term objective of engineering higher-yield strains. The sample size $n=71$ and the number of predictors $p=4088$, where each predictor is the
(log-)expression level of a gene. The response is the log riboflavin production rate.

\begin{figure}[ht]
    \centering
    \includegraphics[width=0.99\linewidth]{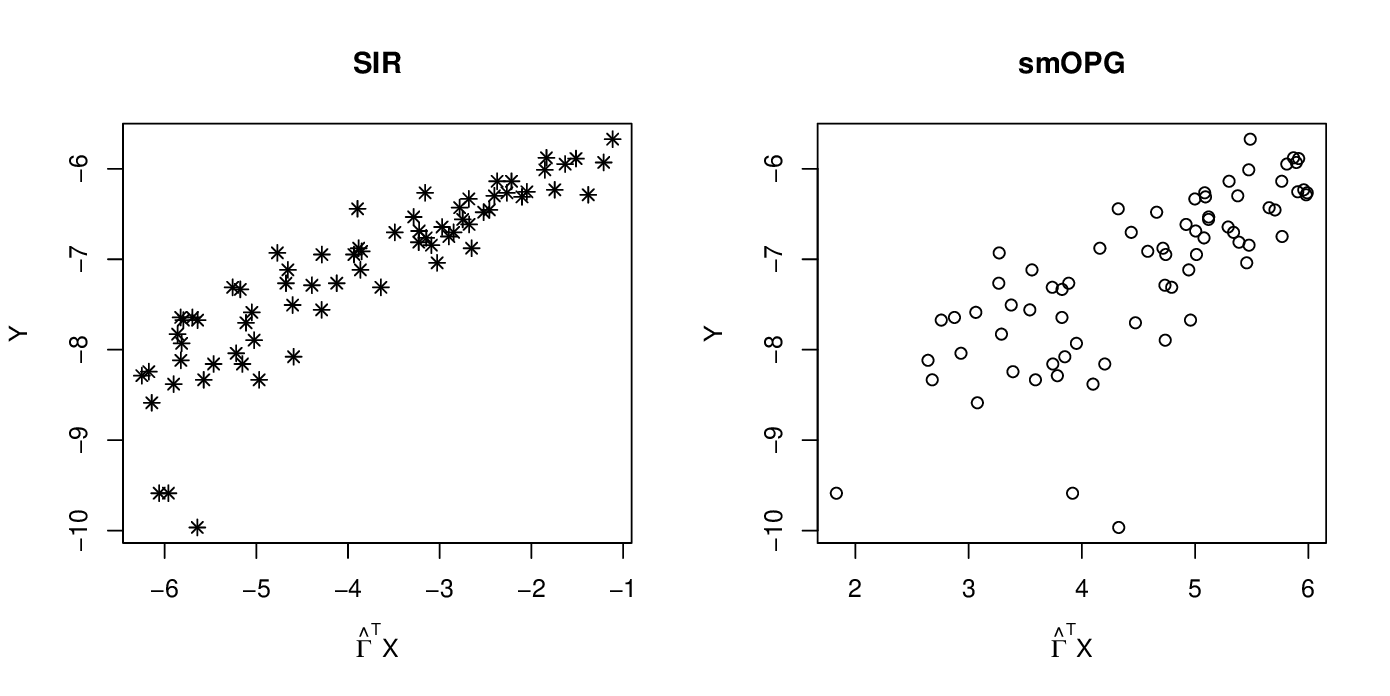}
    \vspace{-0.2in}
    \caption{Estimated sufficient summary plot based on sparse SIR and smOPG}
    \label{fig:essp_ribo}
\end{figure}

We begin by examining the estimated sufficient summary plots based on sparse SIR and smOPG, shown in Figure~\ref{ribo}. Both plots suggest a roughly linear relationship between the response and the fitted index, although three observations appear to depart noticeably from the main trend. The extent to which these observations are separated from the bulk of the data differs somewhat between the two methods, reflecting differences in the estimated directions.

We next compare the competing methods under repeated random training/test partitions, excluding CQR and SMAVE because of computational or implementation limitations in this setting. For each replicate, we randomly select 50 observations for training and use the remaining 21 observations for testing. The predictors are standardized prior to model fitting, and, unless otherwise noted, the tuning settings are the same as those used in the simulation study in Section~\ref{subsec:subset_recovery}.

\begin{table}[htb!]
\small
\caption{Top 10 selected genes (100 splits) with average test RMSE and distance correlation. Genes selected by at all 4 methods are highlighted in magenta, 3 methods are in red, while those highlighted in blue are selected by two methods. Genes in black color were selected by only one method.}
\label{ribo}
\resizebox{\textwidth}{!}{
\begin{tabular}{ccccccccc}
\hline
& \multicolumn{2}{c}{OLS} & \multicolumn{2}{c}{SIR}  & \multicolumn{2}{c}{RANK}  & \multicolumn{2}{c}{smOPG} \\ \hline
No. & Gene & Freq & Gene & Freq & Gene & Freq & Gene & Freq \\ \hline
1  & {\color[HTML]{FF00FF} 2564} & 0.89 & {\color[HTML]{0000FF} 1588} &  0.6  &  3311 &  0.85 &  {\color[HTML]{FF0000} 4003} &  0.81 \\
2 & {\color[HTML]{0000FF} 1762} &  0.82 &  {\color[HTML]{FF00FF} 4004} &  0.55 &  {\color[HTML]{FF0000} 1123} &  0.84 &  {\color[HTML]{0000FF} 1516} &  0.78 \\
3 & {\color[HTML]{0000FF} 624}  &  0.75 & 827 &  0.47 &  {\color[HTML]{0000FF} 1762} & 0.7 & {\color[HTML]{FF00FF} 4004} &  0.51 \\
4 & {\color[HTML]{FF0000} 4003} & 0.73 & {\color[HTML]{FF0000} 1297} & 0.46 & {\color[HTML]{FF00FF} 4004} & 0.7  & {\color[HTML]{FF00FF} 2564} & 0.45 \\
5 & {\color[HTML]{0000FF} 1827} & 0.71 & {\color[HTML]{0000FF} 2184} & 0.45 & 1503  & 0.66 & {\color[HTML]{FF0000} 1123} & 0.44 \\
6 & 73 & 0.68 & {\color[HTML]{FF0000} 1123} & 0.44 & {\color[HTML]{0000FF} 2184} & 0.63 & {\color[HTML]{FF0000} 1297} & 0.41 \\
7 &  2027  & 0.62 & {\color[HTML]{0000FF} 1516} & 0.39 & {\color[HTML]{FF00FF} 2564} & 0.63 & 1312 & 0.36 \\
8 & {\color[HTML]{FF00FF} 4004} & 0.61 & 1585 & 0.35 & {\color[HTML]{FF0000} 1297} & 0.62 & {\color[HTML]{0000FF} 624}  & 0.32 \\
9 & {\color[HTML]{0000FF} 1639} & 0.59 & 1478 & 0.29 & {\color[HTML]{FF0000} 4003} & 0.61 & {\color[HTML]{0000FF} 1588} & 0.25 \\
10 & 1131 & 0.56 & {\color[HTML]{FF00FF} 2564} &  0.29 & {\color[HTML]{0000FF} 1827} &  0.6  & {\color[HTML]{0000FF} 1639} &  0.25 \\ \hline
$\hat s$ &  & 30.68 & & 17.96  & & 34.84 &  & 7.42 \\
$dCor(Y, \X_{\hat{\mathcal A}})$ &  & 0.73 & & 0.70 & & \textbf{0.74}  &  & \textbf{0.74}  \\
RMSE  &  & \textbf{0.54} & & 0.71 &  & 1.16  &  & 0.59 \\ \hline
\end{tabular}
}
\end{table}

Table~\ref{ribo} reports the ten most frequently selected genes for each method, together with the average test RMSE and the average distance correlation. The four procedures produce partially overlapping, but not identical sets of frequently selected genes. Some genes appear repeatedly across multiple methods, suggesting signals that are robust to the choice of regularization strategy. Although Rank, OLS, and smOPG yield similar post-selection distance correlations, smOPG produces the most parsimonious selected gene set while remaining competitive in terms of prediction error. Moreover, nearly all genes frequently selected by smOPG are also selected by at least one other method, suggesting that the smOPG subset reflects a relatively stable selection pattern.

Overall, the results suggest that smOPG achieves a favorable balance between parsimony, predictive performance, and retention of dependence with the response. This feature is particularly appealing in genomics applications, where the primary objective is often to identify a smaller set of candidate genes for downstream experimental validation rather than to optimize prediction alone.

\section{Conclusion}\label{sec:conclusion}

This paper studied active subset recovery in sparse single index models under various contamination mechanisms. To address the limitations of mean- and rank-based procedures in such settings, we proposed smOPG, a sparse median outer-product-of-gradients method that combines local \(\ell_1\)-penalized median regression with sparse rank-one aggregation. The proposed procedure targets the conditional median, avoids reliance on global response ordering, and exploits sparsity at both the local and global stages.

At the population level, we showed that the local median gradient lies in the central subspace, which provides the basis for recovering a sparse sufficient direction through outer-product aggregation. At the sample level, we established a convergence rate for the local penalized median slope estimator and gave a high-level support recovery result for the aggregated sparse direction. The numerical studies indicate that the proposed method remains competitive in clean settings and is particularly effective under contamination, where it tends to provide a favorable balance between parsimonious selection, dependence retention, and predictive performance. The real data analyses further illustrate its practical value for stable variable selection in applications where robustness and interpretability are both important.

The current development focuses on the single index setting. Extensions to multi-index models, grouped or structured sparsity, and more adaptive tuning strategies would be natural directions for future research.


\bibliographystyle{plain}  
\bibliography{references}  

\appendix
\section*{Appendix of Proofs}
\begin{proof}[Proof of Theorem \ref{thm:sparse_median}]\quad\\
Under Assumption \ref{ass_mederror}, the error satisfies
$Q_{1/2}(\epsilon\mid \X)=0$ a.s. Therefore, for each \(\x\in \mathrm{Supp}(\X)\),
\[
m(\x) = Q_{1/2}(Y\mid \X=\x) = Q_{1/2}\{g(\bbeta^\top \x)+\epsilon \mid \X=\x\} = g(\bbeta^\top \x).
\]
By Assumption \ref{ass_diff}, the conditional median function \(m(\x)\) is continuously differentiable on \(\mathrm{Supp}(\X)\). Hence, by the chain rule, $\nabla m(\x) = g'(\bbeta^\top \x)\bbeta$, which implies that \(\nabla m(\x)\in \mathrm{span}(\bbeta)=\spc\) for every \(\x\in \mathrm{Supp}(\X)\). In particular, $\nabla m(\x_0)=g'(\bbeta^\top \x_0)\bbeta$.

Now, let \((a_0,\bb_0)\) denote a minimizer of the population local median criterion
\[
\mathcal L_h(a,\bb;\x_0) = \E\left[
K_h(\X-\x_0)\,
\bigl|Y-a-\bb^\top(\X-\x_0)\bigr|
\right].
\]

By standard local linear M-estimation theory for conditional quantile regression, under Assumptions \ref{ass_mederror}--\ref{ass_diff} and the regularity conditions on the kernel, the slope component of the population local median minimizer satisfies 
\[ 
\bb_0=\nabla m(\x_0)+o(1) \text{ as } h\to 0. 
\]
Since \(\nabla m(\x_0)\in \spc\), we obtain
\[
\mathrm{dist}\{\bb_0,\spc\} = \inf_{\bv\in \spc}\|\bb_0-\bv\|_2 \to 0.
\]
If, in addition, \(g'(\bbeta^\top \x_0)\neq 0\), then \(\nabla m(\x_0)\neq 0\) and $\mathrm{span}\{\nabla m(\x_0)\} = \mathrm{span}(\bbeta)
= \spc$. Hence the angle between \(\mathrm{span}(\bb_0)\) and \(\spc\) converges to zero as \(h\to 0\).

Finally, suppose that \(\bbeta\) is \(s\)-sparse with active set $\mathcal A=\{j:\beta_j\neq 0\}$. Since $\nabla m(\x_0)=g'(\bbeta^\top \x_0)\bbeta$,
we have
\[
\mathrm{supp}\{\nabla m(\x_0)\}\subseteq \mathcal A,
\]
with equality whenever \(g'(\bbeta^\top \x_0)\neq 0\). Therefore the limiting median gradient is supported on the active coordinates and lies in the one-dimensional sparse central direction. If, in addition, the penalized population minimizer \(\bb_{0,\lambda}\) satisfies
\[
\mathrm{supp}(\bb_{0,\lambda})=\mathcal A
\]
for all sufficiently small \(\lambda=\lambda_h\to 0\), then \(\mathrm{span}(\bb_{0,\lambda})\) coincides asymptotically with the sparse central direction \(\spc^{(s)}\).
\end{proof}


\begin{proof}[Proof of Theorem \ref{thm:beta_consistency}]\qquad \\
Fix an interior point \(\x_0\in\mathrm{Supp}(\X)\), and, for notational simplicity, write
\[
\bb_0=\bb_0(\x_0),\qquad
\bb_h=\bb_h(\x_0),\qquad
\hat{\bb}=\hat{\bb}(\x_0),\qquad
\Z_i=\X_i-\x_0.
\]
Let \((a_h,\bb_h)\) denote the minimizer of the population local median criterion at bandwidth \(h\).

Under Assumptions \ref{ass_mederror} and \ref{ass_diff}, the conditional median function \(m(\cdot)\) is continuously differentiable, and by Condition \textup{(C2)} it has bounded derivatives up to order three on \(\mathrm{Supp}(\X)\). Together with the symmetry of the kernel in Condition \textup{(C3)}, standard local linear M-estimation arguments in \cite{fan2016multivariate} imply that the population local slope has bias of order \(h^2\), namely
\begin{align}
\|\bb_h-\bb_0\|_2=O(h^2).
\label{eq:bias_local}
\end{align}

It remains to bound \(\|\hat{\bb}-\bb_h\|_2\). Define
\[
\psi_{1/2}(u)=\frac12-\mathbb I(u<0).
\]
The empirical score for the slope component, evaluated at the population target, is
\begin{align}
S_n
=
\frac{1}{n}\sum_{i=1}^n K_h(\Z_i)\Z_i\psi_{1/2}(\epsilon_i),
\label{eq:score}
\end{align}
where \(K_h(\Z_i)=h^{-p}K(\Z_i/h)\). Since \(Q_{1/2}(\epsilon_i\mid \X_i)=0\), we have $\E\{\psi_{1/2}(\epsilon_i)\mid \X_i\}=0$,
and therefore \(\E(S_n)=0\).

To determine the stochastic order of \(S_n\), observe that
\[
\E\!\left[K_h(\Z)^2\Z\Z^\top \psi_{1/2}(\epsilon)^2\right]
=
h^{2-p}
\int K(\bu)^2\bu\bu^\top
\E\!\left\{\psi_{1/2}(\epsilon)^2\mid \X=\x_0+h\bu\right\}
f_{\X}(\x_0+h\bu)\,d\bu,
\]
where \(\bu=(\X-\x_0)/h\). Because \(|\psi_{1/2}(\epsilon)|\le 1/2\), and Conditions \textup{(C1)}--\textup{(C3)} ensure boundedness of the kernel and the design density, each coordinate of \(S_n\) has variance of order \(h^{2-p}/n\). A Bernstein inequality together with a union bound over the \(p\) coordinates then yields
\begin{align}
\|S_n\|_\infty
=
O_p\left(
\sqrt{\frac{h^{2-p}\log p}{n}}
\right).
\label{eq:score_rate}
\end{align}

Next consider the local Hessian of the population median criterion. Although the absolute-deviation loss is not twice differentiable, the population criterion admits a local quadratic expansion by the standard quantile-regression identity. The corresponding slope Hessian matrix is
\begin{align}
H_h(\x_0)
=
\E\left[
K_h(\Z)\Z\Z^\top
f_{Y\mid \X}\bigl(m(\x_0)\mid \X\bigr)
\right].
\label{eq:curvature}
\end{align}
By Assumption \ref{ass_diff}, the conditional density at the median is uniformly bounded away from zero and infinity, so Condition \textup{(C6)} transfers to \(H_h(\x_0)\) up to multiplicative constants. Moreover, a change of variables shows that
\begin{align}
H_h(\x_0)=h^2H_0(\x_0)+o(h^2),
\label{eq:hessian_rate}
\end{align}
where \(H_0(\x_0)\) is positive definite on the relevant sparse cone. Hence the effective local Hessian is of order \(h^2\).

Choose the tuning parameter so that it dominates the empirical slope score, i.e.,
\begin{align}
\lambda_n
\asymp
\sqrt{\frac{h^{2-p}\log p}{n}}.
\label{eq:lambda_rate}
\end{align}
By the basic inequality for \(\ell_1\)-penalized quantile regression, the error vector \(\hat{\bb}-\bb_h\) belongs to the usual sparse cone. Combining this cone condition with the local restricted eigenvalue condition in \textup{(C6)} and the curvature bound in \eqref{eq:hessian_rate} gives
\begin{align}
\|\hat{\bb}-\bb_h\|_2
=
O_p\left(\frac{\sqrt{s}\lambda_n}{h^2}\right)
=
O_p\left\{
\sqrt{\frac{s\log p}{nh^{p+2}}}
\right\}.
\label{eq:stochastic_rate}
\end{align}
Condition \textup{(C4)} guarantees that this stochastic term converges to zero.

Finally, combining \eqref{eq:bias_local} and \eqref{eq:stochastic_rate} yields
\[
\|\hat{\bb}-\bb_0\|_2
\le
\|\hat{\bb}-\bb_h\|_2+\|\bb_h-\bb_0\|_2
=
O_p\left\{
h^2+\sqrt{\frac{s\log p}{nh^{p+2}}}
\right\},
\]
which establishes \eqref{eq:local_rate}.
\end{proof}

\begin{proof}[Proof of Proposition \ref{prop:aggregate_selection}]\quad\\
From $\widetilde{\B}_n=d_{1n}\bu_1\bbeta^\top+\bm\varepsilon_n$, we obtain
\[
\widetilde{\B}_n^\top \hat{\bu}_1
=
d_{1n}(\bu_1^\top \hat{\bu}_1)\bbeta
+ \bm\varepsilon_n^\top \hat{\bu}_1.
\]
Let \(c_n=\bu_1^\top \hat{\bu}_1\), so that \(|c_n|\overset{p}{\to}1\). For \(j\in\mathcal A^c\), \(\beta_j=0\), and therefore
\[
\left|(\widetilde{\B}_n^\top \hat{\bu}_1)_j\right|
=
\left|(\bm\varepsilon_n^\top \hat{\bu}_1)_j\right|
\le
\|\bm\varepsilon_n^\top \hat{\bu}_1\|_\infty
=
O_p(r_n).
\]
Since \(r_n=o(\xi_n)\), inactive coordinates are asymptotically dominated by the threshold level.

For \(j\in\mathcal A\),
\[
\left|(\widetilde{\B}_n^\top \hat{\bu}_1)_j\right|
\ge
|c_n|d_{1n}|\beta_j|
-
\|\bm\varepsilon_n^\top \hat{\bu}_1\|_\infty.
\]
Because
\[
\xi_n=o\!\left(d_{1n}\min_{j\in\mathcal A}|\beta_j|\right)
\quad\text{and}\quad
\|\bm\varepsilon_n^\top \hat{\bu}_1\|_\infty=o_p(\xi_n),
\]
active coordinates remain asymptotically above the threshold level. Hence $\Pr(\widehat{\mathcal A}=\mathcal A)\to 1$.
\end{proof}

\section*{Additional Simulation Results}
Under the simulation settings described in Section \ref{subsec:subset_recovery}, we consider two additional scenarios designed to assess performance under a symmetric link function and a low signal-to-noise ratio. Specifically, the response is generated as
\begin{enumerate}
\item[IV.] \(y_i = 2.5|1 - \bbeta^\top\x_i| +0.8\epsilon_i\),
\item[V.] \(y_i = \frac{1}{2}\bbeta^\top\x_i + (\bbeta^\top\x_i)\epsilon_i \),
\end{enumerate}
where \(\epsilon_i\sim N(0,1)\). The tuning parameters are selected using the same settings as those in Section \ref{subsec:subset_recovery}.

\begin{table}[htb]
\caption{Mean (standard deviation) of the active set size $\hat s$, true positive rate (TPR), and false positive rate (FPR) for $(n,p) = (50, 10)$ based on 100 random samples. }
\label{tab:sims_additional}
\resizebox{\textwidth}{!}{
\begin{tabular}{cl|ccc|ccc|ccc}
\hline
\multirow{2}{*}{Model} & \multicolumn{1}{c|}{\multirow{2}{*}{Method}} & \multicolumn{3}{c|}{Predictor contamination only} & \multicolumn{3}{c|}{Response contamination only} & \multicolumn{3}{c}{Combined contamination} \\ \cline{3-11} 
& \multicolumn{1}{c|}{}  & $\hat s$ & TPR & FPR  & $\hat s$ & TPR & FPR & $\hat s$  & TPR  & FPR         \\ \hline
\multirow{6}{*}{IV}    & CQR                                          & 10.00 (0.00)    & 1.00 (0.00)    & 1.00 (0.00)    & 10.00 (0.00)    & 1.00 (0.00)    & 1.00 (0.00)   & 10.00 (0.00)  & 1.00 (0.00)  & 1.00 (0.00) \\
& OLS  & 3.86 (2.22)     & 0.41 (0.25)    & 0.36 (0.27)    & 3.72 (2.41)     & 0.40 (0.29)    & 0.35 (0.27)   & 3.72 (2.41)   & 0.40 (0.29)  & 0.35 (0.27) \\
& RANK & 3.22 (2.14)     & 0.36 (0.23)    & 0.29 (0.26)    & 3.25 (2.16)     & 0.37 (0.25)    & 0.28 (0.24)   & 3.25 (2.16)   & 0.37 (0.25)  & 0.28 (0.24) \\
& SIR & {\bf 4.97} (2.11) & 0.51 (0.25)    & 0.49 (0.25)    & {\bf 4.89} (2.03) & 0.49 (0.25)    & 0.48 (0.24)   & {\bf 4.89} (2.03)   & 0.49 (0.25)  & 0.48 (0.24) \\
& SMAVE & 9.61 (0.77)  & 1.00 (0.00) & 0.92 (0.15) & 9.58 (0.66)     & 1.00 (0.00)    & 0.92 (0.13)   & 9.58 (0.66)   & 1.00 (0.00)  & 0.92 (0.13) \\
& smOPG  & {\bf 4.87} (2.39) & 0.56 (0.28) & 0.41 (0.31) & {\bf 4.82} (2.34)     & 0.58 (0.26) & 0.38 (0.31) & {\bf 4.82} (2.34) & 0.58 (0.26)  & 0.38 (0.31) \\ \hline
\multirow{6}{*}{V} & CQR  & 10.00 (0.00)    & 1.00 (0.00)    & 1.00 (0.00)    & 10.00 (0.00)    & 1.00 (0.00)    & 1.00 (0.00)   & 10.00 (0.00)  & 1.00 (0.00)  & 1.00 (0.00) \\
& OLS  & 3.36 (3.14)     & 0.40 (0.38)    & 0.27 (0.31)    & 2.30 (2.80)     & 0.26 (0.33)    & 0.20 (0.29)   & 2.30 (2.80)   & 0.26 (0.33)  & 0.20 (0.29) \\
& RANK  & 3.33 (2.92)     & 0.40 (0.33)    & 0.27 (0.30)    & 3.44 (2.95)     & 0.40 (0.35)    & 0.28 (0.28)   & 3.44 (2.95)   & 0.40 (0.35)  & 0.28 (0.28) \\
& SIR  & 3.64 (1.98)     & 0.38 (0.23)    & 0.35 (0.26)    & 3.96 (1.82)     & 0.40 (0.24)    & 0.39 (0.23)   & 3.96 (1.82)   & 0.40 (0.24)  & 0.39 (0.23) \\
& SMAVE  & 9.60 (0.64)     & 0.97 (0.08)    & 0.95 (0.10)    & 9.67 (0.60)     & 0.97 (0.08)    & 0.96 (0.08)   & 9.67 (0.60)   & 0.97 (0.08)  & 0.96 (0.08) \\
& smOPG & {\bf 4.17} (2.09)  & 0.54 (0.29)    & 0.29 (0.24)    & {\bf 4.43} (2.34)     & 0.57 (0.31)    & 0.32 (0.25)   & {\bf 4.43} (2.34)   & 0.57 (0.31)  & 0.32 (0.25) \\ \hline
\end{tabular}
}
\end{table}

While the performance of all methods deteriorates under both settings, 
Figure~\ref{fig:sims_additional} shows that only the adaptive methods, SMAVE and smOPG, achieve selection performance better than random guessing in Model IV, which features a symmetric link function. In Model V, where the scale of the random fluctuation varies with the index \(\bbeta^\top\X\), only smOPG and RANK outperform random guessing. In both scenarios, smOPG remains conservative, with slightly better precision than the competing methods.

\begin{figure}[htb!]
    \centering
    \includegraphics[width=.9\linewidth]{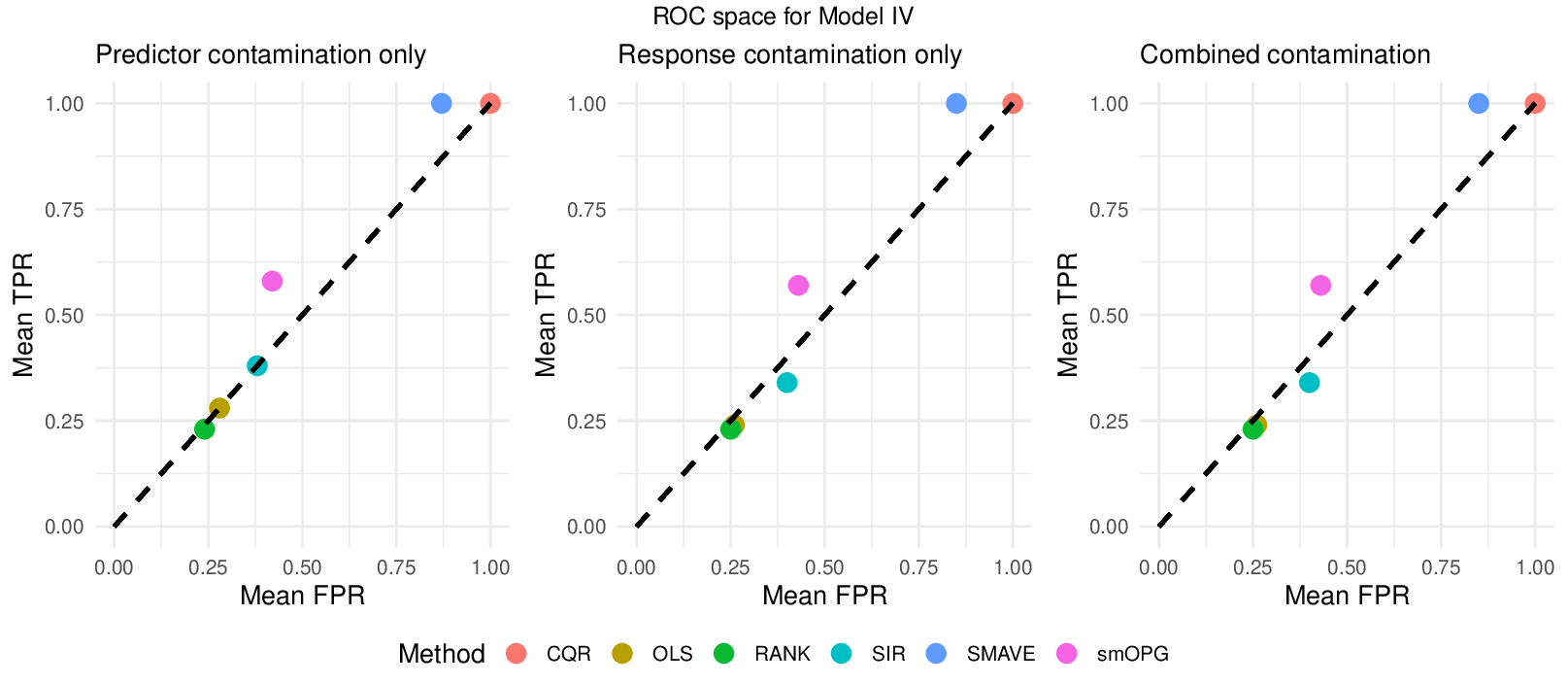}
    \includegraphics[width=.9\linewidth]{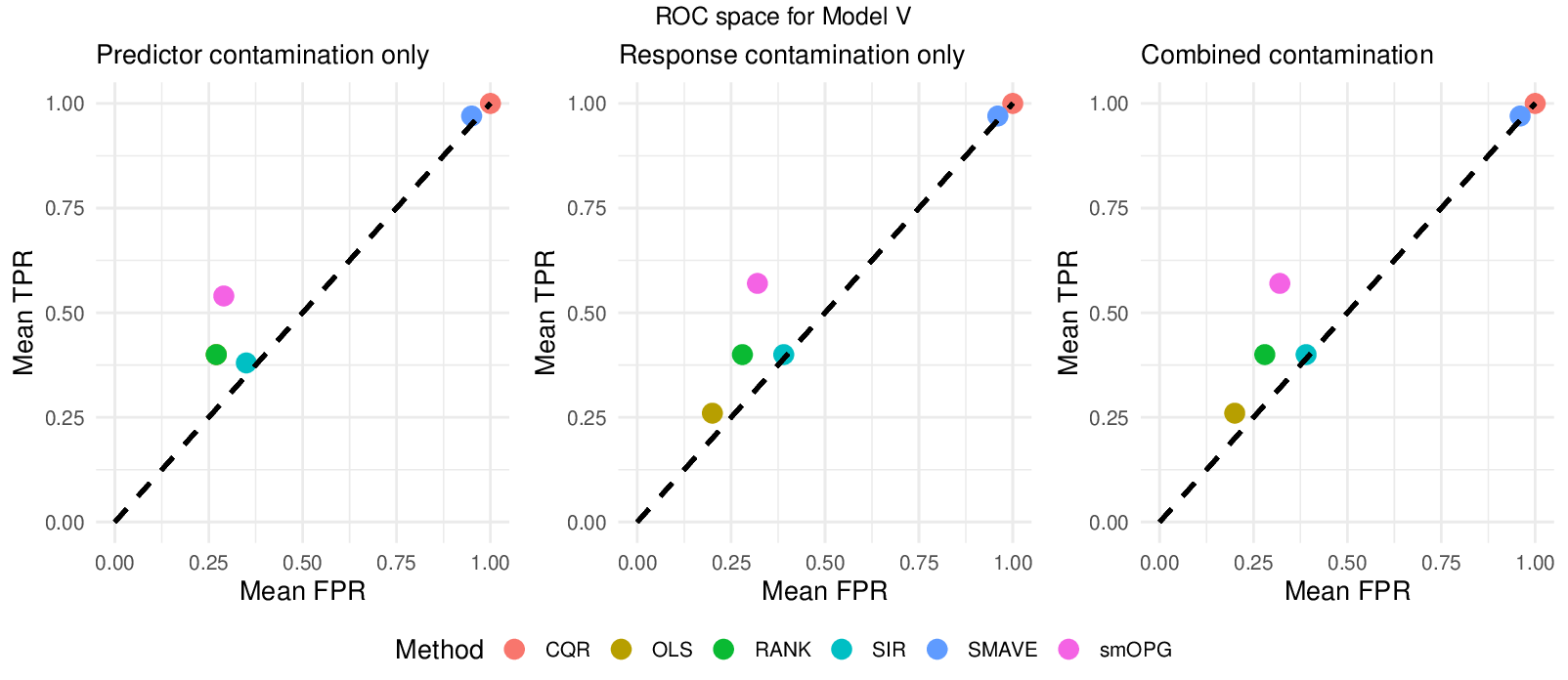}
    \caption{Average TPR vs FPR based on $(n,p)=(50, 10)$ with $5\%$ contamination for Model IV (top) and  Model V (bottom) based on 100 random samples.}
    \label{fig:sims_additional}
\end{figure}
\clearpage


\begin{algorithm}[ht!]
\caption{smOPG-LASSO}
\label{alg:slg}
\begin{algorithmic}[1]

\State Standardize each predictor and compute $\hat{\gamma} =\binom{n}{2}^{-1} \sum_{1\le i<j\le n}\|\x_i-\x_j\|_2$.

\State Estimate the local intercept and slope by solving
    \[
    (\hat a_i,\hat{\bb}_i)
    =
    \arg\min_{a,\bb}
    \left\{
    \sum_{j=1}^n
    \hat w_{ij}\,
    |y_j-a-\bb^\top(\x_j-\x_i)|
    +
    \lambda_n\|\bb\|_1
    \right\}, i,j = 1,\ldots n,
    \]
where $
    \hat w_{ij}  =
    \frac{\exp\{-\|\x_j-\x_i\|_2^2/\hat\gamma^2\}}
    {\sum_{k=1}^n \exp\{-\|\x_k-\x_i\|_2^2/\hat\gamma^2\}}$

\State Define the matrix of local slope estimates as $\widetilde{\B}_n=(\hat{\bb}_1,\ldots,\hat{\bb}_n)^\top$.

\If{\(s\) is prespecified}

    \State Compute the sparse rank-one SVD of \(\widetilde{\B}_n\) with \( \|\bm v_1\|_0 = s\) nonzero right loadings.

\Else

    \State Set \(s\gets 1\) and choose a minimum contribution threshold \(\varrho>0\).

    \State Compute the leading sparse singular value \(\hat d_1(s)\).

    \While{\(s<p\)}

        \State Compute \(\hat d_1(s+1)\).

        \State Set $\delta_s =  \frac{\hat d_1(s+1)-\hat d_1(s)}{\hat d_1(s)}$.
    
        \If{\(\delta_s<\varrho\)}  \textbf{break}
        \EndIf

        \State Update \(s\gets s+1\).

    \EndWhile

\EndIf

\State Let \(\hat{\bEta}_1\) denote the resulting sparse loading vector and set
$\hat{\Gamma} =\hat{\bEta}_1/\|\hat{\bEta}_1\|_2$.

\State Return $\widehat{\mathcal A}=\{j:\hat{\Gamma}_j\neq 0\}$.

\end{algorithmic}
\end{algorithm}

\end{document}